\documentclass[review]{elsarticle}

\usepackage{lineno,hyperref}
\modulolinenumbers[5]

\usepackage{graphicx}
\usepackage{url}

\usepackage{footmisc}
\usepackage{microtype}
\usepackage{booktabs} 
\usepackage{xspace}
\usepackage{color}
\usepackage{times}
\usepackage{appendix}
\usepackage{amsmath,amssymb,latexsym} 
\usepackage{algcompatible}
\usepackage{algorithm}
\usepackage{algpseudocode}
\algrenewcommand\alglinenumber[1]{\tiny #1:}
\usepackage{listings}
\usepackage{verbatim}
\usepackage{multicol}
\usepackage{float}
\usepackage{pifont}
\usepackage{lineno}
\usepackage{array}
\usepackage{soul}
\usepackage{mathtools}
\usepackage{hyphenat}
\usepackage{setspace}
\usepackage{tcolorbox}
\usepackage{wrapfig}
\usepackage{lipsum}
\usepackage{layout}
\usepackage[english]{babel}
\usepackage[utf8]{inputenc}
\usepackage{fancyhdr}
\usepackage{tabularx}
\usepackage{caption}
\usepackage{ragged2e}
\usepackage{hyperref}
\usepackage{cleveref}
\usepackage{fancyvrb}
\usepackage{epstopdf}
\usepackage{epsfig}
\usepackage{etoolbox}
\usepackage{breqn}
\usepackage{relsize}
\usepackage{pgfplots}
\usepackage{subfiles}
\usepackage{pdfpages}
\pgfplotsset{compat=1.7}
\usepackage{calligra}
\usepackage{calrsfs}
\usepackage{subfig}
\usepackage{todonotes}

\newcommand{\blank}[1]{\hspace*{#1}}

\newcommand{\punt}[1]{}
\newcommand{\cmnt}[1]{}

\newtheorem{theorem}{Theorem}

\newtheorem{lemma}[theorem]{Lemma}

\newtheorem{property}[theorem]{Property}

\newtheorem{definition}{Definition}

\newcounter{history}

\newcommand{\secref}[1]{Section~\ref{sec:#1}}
\newcommand{\figref}[1]{Figure~\ref{fig:#1}}

\newcommand{\thmref}[1]{Theorem~\ref{thm:#1}}
\newcommand{\lemref}[1]{Lemma~\ref{lem:#1}}

\newcommand{\defref}[1]{Definition~\ref{def:#1}}

\newcommand{\propref}[1]{Property~\ref{prop:#1}}

\newcommand{\linref}[1]{Line~\ref{lin:#1}}
\newcommand{\algoref}[1]{{Algorithm~\ref{algo:#1}}}

\newcommand{\Lineref}[1]{Line~\ref{lin:#1}}

\newcommand{\ignore}[1]{}

\newcommand{\tobj} {t-object\xspace}
\newcommand{\txns}[1] {txns(#1)}
\newcommand {\comm}[1] {committed(#1)}
\newcommand {\aborted}[1] {aborted(#1)}

\newcommand {\live}[1] {live(#1)}
\newcommand {\term}[1] {term(#1)}

\newcommand{\tseq} {t-sequential\xspace}

\newcommand{\lastw} {lastWrite}

\newcommand{\lupdt}[2] {#2.lastUpdt(#1)}

\newcommand{\mr} {MR}
\newcommand{\tr} {TR}

\newcommand{\validity} {validity\xspace}
\newcommand{\legal} {legal\xspace}
\newcommand{\legality} {legality\xspace}

\newcommand{\op} {operation\xspace}
\newcommand{\mth} {method\xspace}

\newcommand{\cc} {correctness-criterion\xspace}
\newcommand{\ccs} {correctness-criteria\xspace}

\newcommand{\inv} {$inv$}
\newcommand{\rsp} {$rsp$}

\newcommand{\evts}[1] {evts(#1)}
\newcommand{\met}[1] {methods(#1)}

\newcommand{\init} {\emph{init}\xspace}
\newcommand{\tbeg} {\emph{t\_begin}\xspace}
\newcommand{\tread} {\emph{t\_read}\xspace}
\newcommand{\twrite} {\emph{t\_write}\xspace}
\newcommand{\tins} {\emph{t\_insert}\xspace}
\newcommand{\tdel} {\emph{t\_delete}\xspace}
\newcommand{\tlook} {\emph{t\_lookup}\xspace}
\newcommand{\tryc} {\emph{tryC}\xspace}
\newcommand{\trya} {\emph{tryA}\xspace}

\newcommand{\up} {\emph{up}}

\newcommand{\opq} {opaque\xspace}
\newcommand{\opty} {opacity\xspace}

\newcommand{\tab} {hash-table\xspace}

\newcommand{\otm} {\textit{OSTM}\xspace}
\newcommand{\sotm} {\textit{SV-OSTM}\xspace}

\newcommand{\rwtm} {\textit{RWSTM}\xspace}

\newcommand{\lsl} {lazyrb\text{-}list\xspace}
\newcommand{\lazy} {lazy-list\xspace}

\newcommand{\rvm} {\emph{rvm}\xspace}

\newcommand{\fevt}[1] {#1.firstEvt}
\newcommand{\levt}[1] {#1.lastEvt}

\newcommand{\fkmth}[3] {#3.firstKeyMth(#1, #2)}
\newcommand{\pkmth}[3] {#3.prevKeyMth(#1, #2)}

\newcommand\tabspace[1][1cm]{\hspace*{#1}}

\newcommand{\txsetst}[1] {L\_txlog.setStatus($L\_txstatus \downarrow$, $ OK \downarrow$)}

\newcommand{\npintv} {\emph{intraTransValidation()}}
\newcommand{\nptc} {\emph{STM tryC()}}
\newcommand{\npins} {\emph{STM insert()}}
\newcommand{\npdel} {\emph{STM delete()}}
\newcommand{\npluk} {\emph{STM lookup()}}
\newcommand{\nplsls} {\emph{list\_lookup()}}

\newcommand{\nplslins} {\emph{list\_Ins()}}

\newcommand{\npcld}{\emph{commonLu\&Del()}}

\newcommand{\rn} {\textcolor{red}{RL}\xspace}
\newcommand{\bn} {\textcolor{blue}{BL}\xspace}

\newcommand{\rc} {\textcolor{red}{currs[0]}}
\newcommand{\bc} {\textcolor{blue}{currs[1]}}
\newcommand{\bp} {\textcolor{blue}{preds[0]}}
\newcommand{\rp} {\textcolor{red}{preds[1]}}

\newcommand{\mvotm} {\textit{MVOSTM}\xspace}

\newcommand{\hmvotm} {\textit{HT-MVOSTM}\xspace}
\newcommand{\lmvotm} {\textit{list-MVOSTM}\xspace}
\newcommand{\hotm} {\textit{HT-OSTM}\xspace}
\newcommand{\kotm} {\textit{KOSTM}\xspace}
\newcommand{\hkotm} {\textit{HT-KOSTM}\xspace}
\newcommand{\lkotm} {\textit{list-KOSTM}\xspace}
\newcommand{\mvotmgc} {\textit{MVOSTM-GC}\xspace}

\newcommand{\llog} {txLog\xspace}

\newcommand{\txlfind} {$L\_txlog.find(L\_t\_id \downarrow, L\_obj\_id \downarrow, L\_key \downarrow, L\_rec \uparrow)$}

\newcommand{\llsopn}[1] {$L\_rec.setOpn(L\_obj\_id \downarrow$, $L\_key \downarrow$, #1)}
\newcommand{\llsval}[1] {$L\_rec.setVal(L\_obj\_id \downarrow$, $ L\_key \downarrow$, #1)}
\newcommand{\llsopst}[1] {$L\_rec.setOpStatus(L\_obj\_id \downarrow$, $ L\_key \downarrow$, #1)}

\newcommand{\llgopn}[1] {$L\_rec.getOpn(L\_obj\_id \downarrow$, $ L\_key \downarrow$)}

\newcommand{\llgval}[1] {$L\_rec.getVal(L\_obj\_id \downarrow$, $ L\_key \downarrow$)}

\newcommand{\llspc} {$L\_rec.setPred\&Curr(L\_obj\_id \downarrow$, $ L\_key \downarrow$, $G\_pred \downarrow$, $G\_curr \downarrow$)}

\newcommand{\cld}{$commonLu\&Del$($L\_t\_id \downarrow, L\_obj\_id \downarrow, L\_key \downarrow, L\_val \uparrow, L\_op\_status \uparrow$)}

\newcommand{\cnt} {G\_cnt}

\newcommand{\opg}[2] {OPG(#1, #2)}
\newcommand{\ord}[1] {ord(#1)}
\newcommand{\ordfn} {ord}

\newcommand{\mv} {mv}
\newcommand{\rvf} {rvf}
\newcommand{\rt} {rt}
\newcommand{\rvmt} {rv\_method\xspace}
\newcommand{\upmt} {upd\_method\xspace}

\newcommand{\locko} {lockOrder}

\newcommand{\valid} {valid}
\newcommand{\seq}[2] {linearize(#2, #1)}
\newcommand{\stl}[3] {#3.stl(#1, #2)}
\newcommand{\lts}[3] {#3.lts(#1, #2)}
\newcommand{\aco} {accessOrder}

\newcommand{\lsls}[1] {list\_lookup($L\_obj\_id \downarrow,  L\_key \downarrow, G\_pred \uparrow, G\_curr \uparrow$)}
\newcommand{\find} {find\_lts}

\newcommand{\lslins}[1] {list\_Ins($G\_pred \downarrow$, $G\_curr \downarrow$, $node \uparrow$)}

\newcommand{\opmvotm} {\textit{OPT-MVOSTM}\xspace}
\newcommand{\ophmvotm} {\textit{OPT-HT-MVOSTM}\xspace}
\newcommand{\oplmvotm} {\textit{OPT-list-MVOSTM}\xspace}

\newcommand{\opkotm} {\textit{OPT-KOSTM}\xspace}
\newcommand{\ophkotm} {\textit{OPT-HT-KOSTM}\xspace}
\newcommand{\oplkotm} {\textit{OPT-list-KOSTM}\xspace}
\newcommand{\opmvotmgc} {\textit{OPT-MVOSTM-GC}\xspace}
\newcommand{\ophmvotmgc} {\textit{OPT-HT-MVOSTM-GC}\xspace}
\newcommand{\oplmvotmgc} {\textit{OPT-list-MVOSTM-GC}\xspace}

\begin{document}

\begin{frontmatter}

\title{An Efficient Approach to Achieve Compositionality using Optimized Multi-Version Object Based Transactional Systems \footnote{ A preliminary version of this paper appeared in 20th International Symposium on Stabilization, Safety, and Security of Distributed Systems (SSS 2018) and awarded with the \textbf{Best Student Paper Award}. A
		poster version of this work received \textbf{Best Poster Award} in
		NETYS 2018.}}

\author{Chirag Juyal$^1$, Sandeep Kulkarni$^2$, Sweta Kumari$^1$, Sathya Peri$^1$ and Archit Somani$^1$\footnote{Author sequence follows the lexical order of last names. All the authors can be contacted at the addresses given above. Archit Somani's phone number: +91 - 7095044601.}}

\address{Department of Computer Science \& Engineering, IIT Hyderabad, India$^1$ \\
	\texttt{(cs17mtech11014, cs15resch01004, sathya\_p, cs15resch01001)@iith.ac.in}  Department of Computer Science, Michigan State University, MI, USA$^2$ \\
	\texttt{sandeep@cse.msu.edu}}

\begin{abstract}
In the modern era of multi-core systems, the main aim is to utilize the cores properly. This utilization can be done by concurrent programming. But developing a flawless and well-organized concurrent program is difficult.  Software Transactional Memory Systems (STMs) are a convenient programming interface which assist the programmer to access the shared memory concurrently without worrying about consistency issues such as priority-inversion, deadlock, livelock, etc. Another important feature that STMs facilitate is compositionality of concurrent programs with great ease. 
It composes different concurrent operations in a single atomic unit by encapsulating them in a transaction.

Many STMs available in the literature execute read/write primitive operations on memory buffers. We represent them as \emph{Read-Write STMs} or \emph{RWSTMs}. Whereas, there exist some STMs (transactional boosting and its variants) which work on higher level operations  such as insert, delete, lookup, etc. on a hash-table. We refer these STMs as Object Based STMs or OSTMs. 

The literature of databases and RWSTMs say that maintaining multiple versions ensures greater concurrency. This motivates us to maintain multiple version at higher level with object semantics and achieves greater concurrency. So, this paper proposes the notion of \emph{Optimized Multi-version Object Based STMs} or \emph{\opmvotm{s}} which encapsulates the idea of multiple versions in OSTMs to harness the greater concurrency efficiently. For efficient memory utilization, we develop two variations of \emph{\opmvotm{s}}. First, \opmvotm with garbage collection (or \opmvotmgc) which uses unbounded versions but performs garbage collection scheme to delete the unwanted versions. Second, finite version \opmvotm (or \opkotm) which maintains at most $K$ versions by replacing the oldest version when $(K+1)^{th}$ version is created by the current transaction.

We propose the \emph{\opmvotm{s}} for hash-table and list objects as \emph{\ophmvotm} and \emph{\oplmvotm} respectively. For memory utilization, we propose two variants of both the algorithms as \ophmvotmgc, \ophkotm and \oplmvotmgc, \oplkotm respectively.  \ophkotm performs best among its variants and
outperforms state-of-the-art hash-table based STMs  (HT-OSTM, ESTM, RWSTM, HT-MVTO, HT-KSTM) by a factor of 3.62, 3.95, 3.44, 2.75, 1.85 for workload W1 (90\% lookup, 8\% insert and 2\% delete),  1.44, 2.36, 4.45, 9.84, 7.42 for workload W2 (50\% lookup,  25\% insert and 25\% delete), and 2.11, 4.05, 7.84, 12.94, 10.70 for workload W3 (10\% lookup,  45\% insert and 45\% delete) respectively. Similarly, \oplkotm performs best among its variants and outperforms state-of-the-art list based STMs (list-OSTM, Trans-list, Boosting-list, NOrec-list, list-MVTO, list-KSTM)  by a factor of  2.56, 25.38, 23.57, 27.44, 13.34, 5.99 for W1, 1.51, 20.54, 24.27, 29.45, 24.89, 19.78 for W2, and 2.91, 32.88, 28.45, 40.89, 173.92, 124.89 for W3 respectively. \emph{\opmvotm{s}} are generic for other data structures as well. We rigorously proved that \opmvotm{s} satisfy \opty and ensure that transaction with lookup only \mth{s} will never return abort while maintaining unbounded versions. 

\end{abstract}


\end{frontmatter}

\section{Introduction}
\label{sec:intro}


Nowadays, multi-core systems are in trend which necessitated the need for concurrent programming to exploit the cores appropriately. Howbeit, developing the correct and efficient concurrent programs is difficult.  Software Transactional Memory Systems (STMs) are a convenient programming interface which assist the programmer to access the shared memory concurrently using multiple threads without worrying about consistency issues such as deadlock, livelock, priority-inversion, etc. STMs facilitate one more feature compositionality of concurrent programs with great ease which makes it more approachable. Different concurrent operations that need to be composed to form a single atomic unit is achieved by encapsulating them in a transaction. In this paper, we discuss various STMs such as read-write STMs (or RWSTMs), object based STMs (or OSTMs) available in the literature along with the benefits of OSTMs over RWSTMs. After that, we motivated from multi-version RWSTMs and propose multi-version object based STMs (or MVOSTMs) \cite{Juyal+:MVOSTM:SSS:2018} which maintain multiple versions and improves the concurrency further. Later, we made a couple of modifications (discussed in \secref{mvdesign}, \secref{pcode}, and \secref{exp}) to optimize the MVOSTMs and propose optimized MVOSTMs (or \emph{\opmvotm{s}}).


\vspace{1mm}
\noindent
\textbf{Read-Write STMs:} There exists a lot of popular STMs in the literature such as ESTM \cite{Felber+:ElasticTrans:2017:jpdc}, NOrec \cite{Dalessandro+:NoRec:PPoPP:2010} which executes read/write operations on \emph{transaction objects} or \emph{\tobj{s}}. We represent these STMs as \emph{Read-Write STMs} or \emph{RWSTMs}. RWSTMs typically export following methods: (1) \tbeg{}: which begins a transaction with a unique identity, (2) \tread{} (or $r$): which reads the value of \tobj from shared memory, (3) \twrite{} (or $w$): which writes the new value to \tobj in its local memory, (4) \tryc{}: which validates the values written to \tobj{s} by the transaction and tries to commit. If all the updates made by the transaction is consistent then updates reflect to the shared memory and transaction returns commit, and (5) \trya{}: which returns abort on any inconsistency. 

\begin{figure}
	\centering
	\captionsetup{justification=centering}
	\centerline{\scalebox{0.4}{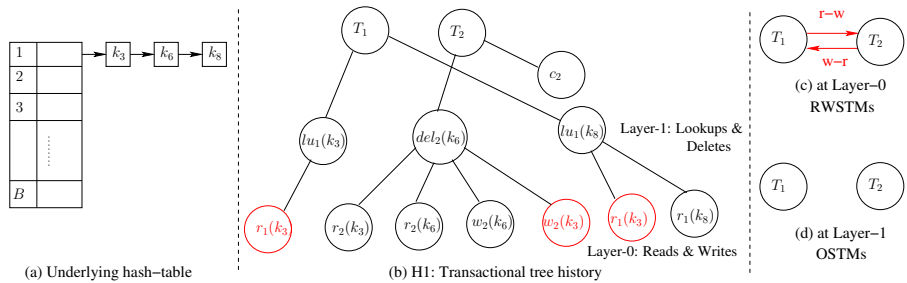}}
	\caption{Advantages of OSTMs over RWSTMs}
	\label{fig:tree-exec}
\end{figure}

 

\vspace{1mm}
\noindent
\textbf{Object based STMs:} There are few STMs available in the literature which executes higher level operations such as insert, delete, lookup on \tab. We represent these STMs as \emph{Object based STMs} or \emph{OSTMs}. The concept of Boosting by Herlihy et al. \cite{HerlihyKosk:Boosting:PPoPP:2008}, the optimistic variant by Hassan et al. \cite{Hassan+:OptBoost:PPoPP:2014} and recently \hotm system by Peri et al. \cite{Peri+:OSTM:Netys:2018} are some examples that demonstrate the performance benefits achieved by \otm{s}. Peri et al. \cite{Peri+:OSTM:Netys:2018} showed that \otm{s} provide greater concurrency than RWSTMs while reducing the number of aborts.

\vspace{1mm}
\noindent 
\textbf{Benefits of \otm{s} over \rwtm{s}: } To show the benefits of \otm{s}, We consider a \tab based STM system which invokes insert (or $ins$), lookup (or $lu$) and delete (or $del$) method. Each \tab consists of $B$ buckets with the elements in each bucket arranged in the form of a linked-list. \figref{tree-exec} (a) represents a \tab{} with the first bucket containing keys $\langle k_3,~ k_6,~ k_8 \rangle$. \figref{tree-exec} (b) shows the execution by two transaction $T_1$ and $T_2$ represented in the form of a tree. $T_1$ performs lookup \op{s} on keys $k_3$ and $k_8$ while $T_2$ performs a delete on $k_6$. The delete on key $k_6$ generates read on the keys $k_3,k_6$ and writes the keys $k_6,k_3$ assuming that delete is performed similar to delete \op in \lazy \cite{Heller+:LazyList:PPL:2007}. The lookup on $k_3$ generates read on $k_3$ while the lookup on $k_8$ generates read on $k_3, k_8$. Note that in this execution $k_6$ has already been deleted by the time lookup on $k_8$ is performed. 


  
In this execution, we denote the read-write \op{s} (leaves) as layer-0 and $lu, del$ methods as layer-1. Consider the history (execution) at layer-0 (while ignoring higher-level \op{s}), denoted as $H0$. It can be verified this history is not \opq \cite{GuerKap:Opacity:PPoPP:2008}. This is because, between the two reads of $k_3$ by $T_1$, $T_2$ writes to $k_3$. It can be seen that if history $H0$ is input to an \rwtm{s} one of the transactions between $T_1$ or $T_2$ would be aborted to ensure opacity \cite{GuerKap:Opacity:PPoPP:2008}. \figref{tree-exec} (c) shows the presence of a cycle in the conflict graph of $H0$. 
  
Now, consider the history $H1$ at layer-1 consists of $lu$, and $del$ \mth{s}, while ignoring the read/write \op{s} since they do not overlap (referred to as pruning in \cite[Chap 6]{WeiVoss:TIS:2002:Morg}). These methods work on distinct keys ($k_3$, $k_6$, and $k_8$). They do not overlap and are not conflicting. So, they can be re-ordered in either way. Thus, $H1$ is \opq{} \cite{GuerKap:Opacity:PPoPP:2008} with equivalent serial history $T_1 T_2$ (or $T_2 T_1$) and the corresponding conflict graph shown in \figref{tree-exec} (d). Hence, a \tab based \otm{} system does not abort any of $T_1$ or $T_2$. This shows that \otm{s} can reduce the number of aborts and provide greater concurrency. 


\ignore{
\color{red}
\vspace{1mm}
\noindent
\textbf{Multi-Version Object STMs:} Having shown the advantage achieved by \otm{s},\todo{issue 1} We now explore the notion of \emph{Multi-Version Object STMs} or \emph{\mvotm{s}}. It was observed in databases and \rwtm{s} that by storing multiple versions for each \tobj, greater concurrency can be obtained \cite{Kumar+:MVTO:ICDCN:2014}. Maintaining multiple versions can ensure that more read operations succeed because the reading \op{} will have an appropriate version to read. This motivated us to develop \mvotm{s}. 

issue1 = We did not show this. The writeup suggests that we did this. 

\color{blue}
}

\vspace{1mm}
\noindent
\textbf{Multi-Version Object Based STMs:} Some of the \otm{s} such as \cite{HerlihyKosk:Boosting:PPoPP:2008}, \cite{Hassan+:OptBoost:PPoPP:2014}, \cite{Peri+:OSTM:Netys:2018} exploits the advantages of it. In this paper, we propose and analyze \emph{Optimized Multi-version Object Based STMs} or \emph{\opmvotm{s}} along with the rigorous correctness proof. This work is motivated by the observation that databases and \rwtm{s} achieves greater concurrency by storing multiple versions corresponding to each \tobj \cite{Kumar+:MVTO:ICDCN:2014}. Specifically, maintaining multiple versions can ensure that more read operations succeed because the reading \op{} will obtain an appropriate version to read. Our goal is to analyze the benefit of \emph{OPT-\mvotm{s}} over both single version \otm{s} and multi-version \rwtm{s}. 

\begin{figure}
	\centering
	\captionsetup{justification=centering}
	\scalebox{.4}{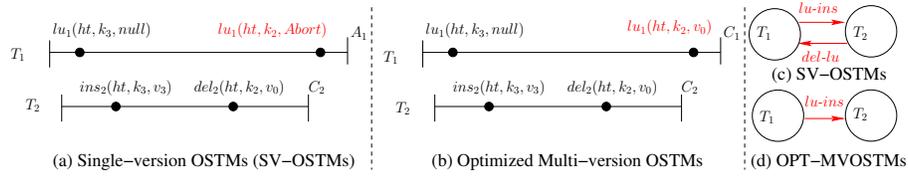}
	\caption{Advantages of optimized multi-version over single version \otm{}}
	\label{fig:pop}
\end{figure}
\noindent 
\textbf{The potential benefit of \emph{OPT-\mvotm{s}} over \otm{s} and multi-version \rwtm{s}:} We now illustrate the advantage of \emph{OPT-\mvotm{s}} as compared to single-version \otm{s} (\sotm{s}) using the \tab object with $B$ buckets having the same \op{s} as discussed above: $ins, lu, del$. \figref{pop} (a) represents a history H with two concurrent transactions $T_1$ and $T_2$ operating on a \tab{} $ht$. $T_1$ first tries to perform a $lu$ on key $k_3$. But due to the absence of key $k_3$ in $ht$, it obtains a value of $null$. Then $T_2$ invokes $ins$ method on the same key $k_3$ and inserts the value $v_3$ in $ht$. Then $T_2$ deletes the key $k_2$ from $ht$ and returns $v_0$ implying that some other transaction had previously inserted $v_0$ into $k_2$. The second method of $T_1$ is $lu$ on the key $k_2$. With this execution, any \sotm system has to return abort for $T_1$'s $lu$ \op  to ensure correctness, i.e., \opty. Otherwise, if $T_1$ would have obtained a return value $v_0$ for $k_2$, then the history would not be \opq anymore. This is reflected by a cycle in the corresponding conflict graph between $T_1$ and $T_2$, as shown in \figref{pop} (c). Thus to ensure \opty, \sotm system has to return abort for $T_1$'s lookup on $k_2$.

 In an \emph{OPT-\mvotm{s}} based on \tab, denoted as \emph{OPT-\hmvotm}, whenever a transaction inserts or deletes a key $k$, a new version is created. Consider the above example with an \emph{OPT-\hmvotm}, as shown in \figref{pop} (b). Even after $T_2$ deletes $k_2$, the previous value of $v_0$ is still retained. Thus, when $T_1$ invokes $lu$ on $k_2$ after the delete on $k_2$ by $T_2$, \emph{OPT-\hmvotm} return $v_0$ (as previous value). With this, the resulting history is \opq{} with equivalent serial history being $T_1 T_2$. The corresponding conflict graph is shown in \figref{pop} (d) does not have a cycle. 

Thus, \emph{OPT-\mvotm} reduces the number of aborts and achieve greater concurrency than \sotm{s} while ensuring the compositionality. We believe that the benefit of \emph{OPT-\mvotm} over multi-version \rwtm is similar to \sotm over single-version \rwtm as explained above.
\emph{OPT-\mvotm} is a generic concept which can be applied to any data structure. In this paper, we have considered the \tab{} and list  based \emph{OPT-\mvotm{s}} as \emph{OPT-\hmvotm} and \emph{OPT-\lmvotm} respectively. If the bucket size $B$ of \tab{} becomes 1 then \tab based \emph{OPT-\mvotm{s}} boils down to the list based \emph{OPT-\mvotm{s}}. 

\emph{OPT-\hmvotm} and \emph{OPT-\lmvotm} use an unbounded number of versions for each key. To address this issue, we develop two variants for both \tab and list data structures (or DS): (1) A garbage collection method in \emph{OPT-\mvotm{s}} to delete the unwanted versions of a key, denoted as \emph{OPT-\mvotmgc}. Garbage collection gave an average performance gain of 16\%  over \emph{OPT-\mvotm} without garbage collection in the best case. Thus, the overhead of garbage collection scheme is less than the performance improvement due to improved memory usage. (2) Placing a limit of $K$ on the number versions in \emph{OPT-\mvotm}, resulting in \emph{OPT-\kotm}. This gave an average performance gain of 24\% over \emph{OPT-\mvotm} without garbage collection in the best case. 

 Experimental results show that \ophkotm performs best among its variants and
outperforms state-of-the-art hash-table based STMs  (HT-OSTM, ESTM, RWSTM, HT-MVTO, HT-KSTM) by a factor of 3.62, 3.95, 3.44, 2.75, 1.85 for workload W1 (90\% lookup, 8\% insert and 2\% delete),  1.44, 2.36, 4.45, 9.84, 7.42 for workload W2 (50\% lookup,  25\% insert and 25\% delete), and 2.11, 4.05, 7.84, 12.94, 10.70 for workload W3 (10\% lookup,  45\% insert and 45\% delete) respectively. Similarly, \oplkotm performs best among its variants and outperforms state-of-the-art list based STMs (list-OSTM, Trans-list, Boosting-list, NOrec-list, list-MVTO, list-KSTM)  by a factor of  2.56, 25.38, 23.57, 27.44, 13.34, 5.99 for W1, 1.51, 20.54, 24.27, 29.45, 24.89, 19.78 for W2, and 2.91, 32.88, 28.45, 40.89, 173.92, 124.89 for W3 respectively. To the best of our knowledge, this is the first work to explore the idea of using multiple versions in \otm{s} to achieve greater concurrency. 


\noindent 
\textbf{Contributions of the paper:}
\begin{itemize}
\item We propose a new notion of optimized multi-version objects based STM system as \emph{OPT-\mvotm} in \secref{mvdesign}. In this paper, we develop it for list and \tab{} objects as \emph{OPT-\lmvotm} and \emph{OPT-\hmvotm} respectively. \emph{OPT-\mvotm} is generic for other data structures as well.

\item For efficient space utilization in \emph{\opmvotm{s}} with unbounded versions, we develop \emph{Garbage Collection} for \emph{OPT-\mvotm} (i.e. \emph{OPT-\mvotmgc}) and bounded version \emph{OPT-\mvotm} (i.e. \emph{OPT-\kotm}).

\item \secref{cmvostm} shows that \emph{OPT-\lmvotm} and \emph{OPT-\hmvotm} satisfy standard \cc of STMs, \emph{opacity} \cite{GuerKap:Opacity:PPoPP:2008}.


\item  Experimental analysis of both \emph{OPT-\lmvotm} and \emph{OPT-\hmvotm} with state-of-the-art STMs are present in \secref{exp}. Proposed \emph{OPT-\lmvotm} and \emph{OPT-\hmvotm} provide greater concurrency and reduces the number of aborts as compared to  \emph{MVOSTMs}, \sotm{s}, single-version \rwtm{s} and, multi-version \rwtm{s} while maintaining multiple versions corresponding to each key. 


\end{itemize}

\noindent
\textit{Roadmap:} The paper is organized as follows. We describe our building system model in \secref{model}. In \secref{gcofo}, we formally define the graph characterization of opacity. \secref{mvdesign} represents the \emph{\opmvotm{s}} design and data structure. \secref{pcode} shows the working of \emph{\ophmvotm{s}} and its algorithms. We formally prove the correctness of \emph{\opmvotm{s}} in \secref{cmvostm}.  In \secref{exp} we show the experimental evaluation of \emph{\opmvotm{s}} with state-of-art-STMs. Finally, we conclude in \secref{confu}. 

\vspace{-4mm}
\section{Building System Model}
\label{sec:model}

 Our assumption follows ~\cite{tm-book,Peri+:OSTM:Netys:2018} in which the system consists of a finite set of $p$ processes, $p_1,\ldots,p_n$, accessed by a finite number of $n$ threads in a completely asynchronous fashion and communicates each other using shared keys (or objects). The threads invoke higher level \mth{s} on the shared objects and get corresponding responses. Consequently, we make no assumption about the relative speeds of the threads. We also assume that none of these processors and threads fail or crash abruptly.
 
\vspace{1mm}
\noindent
\textbf{Events and Methods:} We assume that the threads execute atomic \emph{events} and the events by different threads are (1) read/write on shared/local memory objects, (2) \mth{} invocations (or \emph{\inv}) event and responses (or \emph{\rsp}) event on higher level shared memory objects.

Within a transaction, a process can invoke layer-1 \mth{s} (or \op{s}) on a \emph{\tab} \tobj. A \tab{}($ht$) consists of multiple key-value pairs of the form $\langle k, v \rangle$. The keys and values are respectively from sets $\mathcal{K}$ and $\mathcal{V}$. The \mth{s} that a thread can invoke are: (1) $\tbeg_i(){}$: begins a transaction and returns a unique id to the invoking thread. (2) $\tins_i(ht, k, v)$: transaction $T_i$ inserts a value $v$ onto key $k$ in $ht$. (3) $\tdel_i(ht, k, v)$: transaction $T_i$ deletes the key $k$ from the \tab{} $ht$ and returns the current value $v$ for $T_i$. If key $k$ does not exist, it returns $null$. (4) $\tlook_i(ht, k, v)$: returns the current value $v$ for key $k$ in $ht$  for $T_i$. Similar to \tdel, if the key $k$ does not exist then \tlook returns $null$. (5) $\tryc_i()$: which tries to commit all the \op{s} of $T_i$  and (6) $\trya_i()$: aborts $T_i$. We assume that each \mth{} consists of an \inv{} and \rsp{} event.


We denote \tins{} and \tdel{} as \emph{update} \mth{s} (or $\upmt{}$ or $up$) since both of these change the underlying data structure. We denote \tdel{} and \tlook{} as \emph{return-value methods (or $\rvmt{}$ or $rvm$)} as these operations return values from $ht$. A \mth{} may return $ok$ if successful or $\mathcal{A}$(abort) if it sees an inconsistent state of $ht$. 

Formally, we denote a \mth{} $m$ by the tuple $\langle \evts{m}, <_m\rangle$. Here, $\evts{m}$ are all the events invoked by $m$ and the $<_m$ a total order among these events.

\vspace{1mm}
\noindent
\textbf{Transactions:} Following the notations used in database multi-level transactions\cite{WeiVoss:TIS:2002:Morg}, we model a transaction as a two-level tree. The \emph{layer-0} consist of read/write events and \emph{layer-1} of the tree consists of \mth{s} invoked by a transaction.

Having informally explained a transaction, we formally define a transaction $T$ as the tuple $\langle \evts{T}, <_T\rangle$. Here $\evts{T}$ are all the read/write events at \emph{layer-0} of the transaction. $<_T$ is a total order among all the events of the transaction.

We denote the first and last events of a transaction $T_i$ as $\fevt{T_i}$ and $\levt{T_i}$. Given any other read/write event $rw$ in $T_i$, we assume that $\fevt{T_i} <_{T_i} rw <_{T_i} \levt{T_i}$. All the \mth{s} of $T_i$ are denoted as $\met{T_i}$. We assume that for any method $m$ in $\met{T_i}$, $\evts{m}$ is a subset of $\evts{T_i}$ and $<_m$ is a subset of $<_{T_i}$. We assume that if a transaction has invoked a \mth, then it does not invoke a new \mth{} until it gets the response of the previous one. Thus all the \mth{s} of a transaction can be ordered by $<_{T_i}$. Formally,  $(\forall m_{p}, m_{q} \in \met{T_i}: (m_{p} <_{T_i} m_{q}) \lor (m_{q} <_{T_i} m_{p}))\rangle$, here $m_{p}$ and $m_{q}$ are $p_{th}$ and $q_{th}$ methods of $T_i$ respectively. 

\noindent
\textbf{Histories:} A \emph{history} is a sequence of events belonging to different transactions. The collection of events is denoted as $\evts{H}$. Similar to a transaction, we denote a history $H$ as tuple $\langle \evts{H},<_H \rangle$ where all the events are totally ordered by $<_H$. The set of \mth{s} that are in $H$ is denoted by $\met{H}$. A \mth{} $m$ is \emph{incomplete} if $\inv(m)$ is in $\evts{H}$ but not its corresponding response event. Otherwise, $m$ is \emph{complete} in $H$. 

Coming to transactions in $H$, the set of transactions in $H$ are denoted as $\txns{H}$. The set of committed (resp., aborted) transactions in $H$ is denoted by $\comm{H}$ (resp., $\aborted{H}$). The set of \emph{live} transactions in $H$ are those which are neither committed nor aborted and denoted as $\live{H} =\txns{H}-\comm{H}-\aborted{H}$.  On the other hand, the set of \emph{terminated} transactions are those which have either committed or aborted and is denoted by $\term{H} = \comm{H} \cup \aborted{H}$.

The relation between the events of transactions \& histories is analogous to the relation between \mth{s} \& transactions. We assume that for any transaction $T$ in $\txns{H}$, $\evts{T}$ is a subset of $\evts{H}$ and $<_T$ is a subset of $<_{H}$. Formally, $\langle \forall T \in \txns{H}: (\evts{T} \subseteq \evts{H}) ~ \land (<_T \subseteq <_{H}) \rangle$. 

We denote two histories $H_1, H_2$ as \emph{equivalent} if their events are the same, i.e., $\evts{H_1} = \evts{H_2}$. A history $H$ is qualified to be \emph{well-formed} if: (1) all the \mth{s} of a transaction $T_i$ in $H$ are totally ordered, i.e. a transaction invokes a \mth{} only after it receives a response of the previous \mth{} invoked by it (2) $T_i$ does not invoke any other \mth{} after it received an $\mathcal{A}$ response or after $\tryc(ok)$ \mth. We only consider \emph{well-formed} histories for \opmvotm.

A \mth{} $m_{ij}$ ($j^{th}$ method of a transaction $T_i$) in a history $H$ is said to be \emph{isolated} or \emph{atomic} if for any other event $e_{pqr}$ ($r^{th}$ event of method $m_{pq}$) belonging to some other \mth{} $m_{pq}$ of transaction $T_p$ either $e_{pqr}$ occurs before $\inv(m_{ij})$ or after $\rsp(m_{ij})$. 

\vspace{1mm}
\noindent
\textbf{Sequential Histories:} A history $H$ is said to be \emph{sequential} (term used in \cite{KuznetsovPeri:Non-interference:TCS:2017, KuznetsovRavi:ConcurrencyTM:OPODIS:2011}) if all the methods in it are complete and isolated. From now onwards, most of our discussion would relate to sequential histories. 

Since in sequential histories all the \mth{s} are isolated, we treat each \mth as a whole without referring to its $inv$ and $rsp$ events. For a sequential history $H$, we construct the \emph{completion} of $H$, denoted $\overline{H}$, by inserting $\trya_k(\mathcal{A})$ immediately after the last \mth{} of every transaction $T_k \in \live{H}$. Since all the \mth{s} in a sequential history are complete, this definition only has to take care of completed transactions. 

Consider a sequential history $H$. Let $m_{ij}(ht, k, v/nil)$ be the first \mth of $T_i$ in $H$ operating on the key $k$ as $\fkmth{\langle ht, k \rangle}{T_i}{H}$, where $m_{ij}$ stands for $j^{th}$ method of $i^{th}$ transaction. For a \mth $m_{ix}(ht, k, v)$ which is not the first \mth on $\langle ht, k \rangle$ of $T_i$ in $H$, we denote its previous \mth on $k$ of $T_i$ as $m_{ij}(ht, k, v) = \pkmth{m_{ix}}{T_i}{H}$.

\vspace{1mm}
\noindent
\textbf{Real-time Order and Serial Histories:} Given a history $H$, $<_H$ orders all the events in $H$. For two complete \mth{s} $m_{ij}, m_{pq}$ in $\met{H}$, we denote $m_{ij} \prec_H^{\mr} m_{pq}$ if $\rsp(m_{ij}) <_H \inv(m_{pq})$. Here \mr{} stands for method real-time order. It must be noted that all the \mth{s} of the same transaction are ordered. Similarly, for two transactions $T_{i}, T_{p}$ in $\term{H}$, we denote $(T_{i} \prec_H^{\tr} T_{p})$ if $(\levt{T_{i}} <_H \fevt{T_{p}})$. Here \tr{} stands for transactional real-time order. 

\cmnt{
Thus, $\prec$ partially orders all the \mth{s} and transactions in $H$. It can be seen that if $H$ is sequential, then $\prec_H^{\mr}$ totally orders all the \mth{s} in $H$. Formally, $\langle (H \text{ is seqential}) \implies (\forall m_{ij}, m_{pq} \in \met{H}: (m_{ij} \prec_H^{\mr} m_{pq}) \lor (m_{pq} \prec_H^{\mr} m_{ij}))\rangle$. 
}

We define a history $H$ as \emph{serial} \cite{Papad:1979:JACM} or \emph{t-sequential} \cite{KuznetsovRavi:ConcurrencyTM:OPODIS:2011} if all the transactions in $H$ have terminated and can be totally ordered w.r.t $\prec_{\tr}$, i.e. all the transactions execute one after the other without any interleaving. Intuitively, a history $H$ is serial if all its transactions can be isolated. Formally, $\langle (H \text{ is serial}) \implies (\forall T_{i} \in \txns{H}: (T_i \in \term{H}) \land (\forall T_{i}, T_{p} \in \txns{H}: (T_{i} \prec_H^{\tr} T_{p}) \lor (T_{p} \prec_H^{\tr} T_{i}))\rangle$. Since all the methods within a transaction are ordered, a serial history is also sequential.

\ignore{
\vspace{1mm}
\noindent
\textbf{Real-time Order \& Serial Histories:}  Two \mth{s} $m_{ij}$ and $m_{pq}$ of history $H$ are in real-time order, if $\rsp(m_{ij}) <_H \inv(m_{pq})$. Similarly, two transactions $T_i$ and $T_j$ are in real-time order, if  $(\levt{T_{i}} <_H \fevt{T_{j}})$, where $\levt{T_{i}}$ and $\fevt{T_{j}}$ represents the last method of $T_i$ and first method of $T_j$ respectively. A history $H$ is said to be serial if all the transactions are atomic and totally ordered.

\vspace{1mm}
\noindent
\textbf{\textit{ Valid and Legal Histories:}} A history $H$ is said to valid if all the \rvmt{s} are lookup from previously committed 
}


\vspace{1mm}
\noindent
\textbf{Valid Histories:} A \rvmt{} (\tdel{} and \tlook{}) $rvm_{ij}$ on key $k$ is valid if it returns the value updated by any of the previously committed transaction that updated key $k$. A history $H$ is said to valid if all the \rvmt{s} of H are valid. 

\vspace{1mm}
\noindent
\textbf{Legal Histories:} 
We define the \emph{\legality{}} of \rvmt{s} on sequential histories which we use to define correctness criterion as opacity \cite{GuerKap:Opacity:PPoPP:2008}. Consider a sequential history $H$ having a \rvmt{} $\rvm_{ij}(ht, k, v)$ (with $v \neq null$) as $j^{th}$ method belonging to transaction $T_i$. We define this \rvm \mth{} to be \emph{\legal} if: 
\vspace{-1mm}
\begin{enumerate}
	\item[Rule 1] \label{step:leg-same} If the $\rvm_{ij}$ is not the first \mth of $T_i$ to operate on $\langle ht, k \rangle$ and $m_{ix}$ is the previous \mth of $T_i$ on $\langle ht, k \rangle$. Formally, $\rvm_{ij} \neq \fkmth{\langle ht, k \rangle}{T_i}{H}$ $\land (m_{ix}(ht, k, v') = \pkmth{\langle ht, k \rangle}{T_i}{H})$ (where $v'$ could be null). Then,
	\begin{enumerate}
		\setlength\itemsep{0em}
		\item If $m_{ix}(ht, k, v')$ is a \tins{} \mth then $v = v'$. 
		\item If $m_{ix}(ht, k, v')$ is a \tlook{} \mth then $v = v'$. 
		\item If $m_{ix}(ht, k, v')$ is a \tdel{} \mth then $v = null$.
	\end{enumerate}
	
	In this case, we denote $m_{ix}$ as the last update \mth{} of $\rvm_{ij}$, i.e., 
	
	 $m_{ix}(ht, k, v') = \lupdt{\rvm_{ij}(ht, k, v)}{H}$. 
	
	\item[Rule 2] \label{step:leg-ins} If $\rvm_{ij}$ is the first \mth{} of $T_i$ to operate on $\langle ht, k \rangle$ and $v$ is not null. Formally, $\rvm_{ij}(ht, k, v) = \fkmth{\langle ht, k \rangle}{T_i}{H} \land (v \neq null)$. Then,
	\begin{enumerate}
		\setlength\itemsep{0em}
		\item There is a \tins{} \mth{} $\tins_{pq}(ht, k, v)$ in $\met{H}$ such that $T_p$ committed before $\rvm_{ij}$. Formally, $\langle \exists \tins_{pq}(ht, k, v) \in \met{H} : \tryc_p \prec_{H}^{\mr} \rvm_{ij} \rangle$. 
		\item There is no other update \mth{} $up_{xy}$ of a transaction $T_x$ operating on $\langle ht, k \rangle$ in $\met{H}$ such that $T_x$ committed after $T_p$ but before $\rvm_{ij}$. Formally, $\langle \nexists up_{xy}(ht, k, v'') \in \met{H} : \tryc_p \prec_{H}^{\mr} \tryc_x \prec_{H}^{\mr} \rvm_{ij} \rangle$. 		
	\end{enumerate}
	
	In this case, we denote $\tryc_{p}$ as the last update \mth{} of $\rvm_{ij}$, i.e.,  $\tryc_{p}(ht, k, v)$= $\lupdt{\rvm_{ij}(ht, k, v)}{H}$.
	
	\item[Rule 3] \label{step:leg-del} If $\rvm_{ij}$ is the first \mth of $T_i$ to operate on $\langle ht, k \rangle$ and $v$ is null. Formally, $\rvm_{ij}(ht, k, v) = \fkmth{\langle ht, k \rangle}{T_i}{H} \land (v = null)$. Then,
	\begin{enumerate}
		\setlength\itemsep{0em}
		\item There is \tdel{} \mth{} $\tdel_{pq}(ht, k, v')$ in $\met{H}$ such that $T_p$ committed before $\rvm_{ij}$. Formally, $\langle \exists \tdel_{pq}\\(ht, k,$ $ v') \in \met{H} : \tryc_p \prec_{H}^{\mr} \rvm_{ij} \rangle$. Here $v'$ could be null. 
		\item There is no other update \mth{} $up_{xy}$ of a transaction $T_x$ operating on $\langle ht, k \rangle$ in $\met{H}$ such that $T_x$ committed after $T_p$ but before $\rvm_{ij}$. Formally, $\langle \nexists up_{xy}(ht, k, v'') \in \met{H} : \tryc_p \prec_{H}^{\mr} \tryc_x \prec_{H}^{\mr} \rvm_{ij} \rangle$. 		
	\end{enumerate}
	In this case, we denote $\tryc_{p}$ as the last update \mth{} of $\rvm_{ij}$, i.e., $\tryc_{p}(ht, k, v)$ $= \lupdt{\rvm_{ij}(ht, k, v)}{H}$. 
\end{enumerate}
We assume that when a transaction $T_i$ operates on key $k$ of a \tab{} $ht$, the result of this \mth is stored in \emph{local logs} of $T_i$, $\llog_i$ for later \mth{s} to reuse. Thus, only the first \rvmt{} operating on $\langle ht, k \rangle$ of $T_i$ accesses the shared memory. The other \rvmt{s} of $T_i$ operating on $\langle ht, k \rangle$ do not access the shared memory and they see the effect of the previous \mth{} from the \emph{local logs}, $\llog_i$. This idea is utilized in Rule 1. With reference to Rule 2 and Rule 3, it is possible that $T_x$ could have aborted before $\rvm_{ij}$. 


Coming to \tins{} \mth{s}, since a \tins{} \mth{} always returns $ok$ as they overwrite the node if already present therefore they always take effect on the $ht$. Thus, we denote all \tins{} \mth{s} as \legal{} and only give legality definition for \rvmt{}. We denote a sequential history $H$ as \emph{\legal} or \emph{linearized} if all its \rvm \mth{s} are \legal. We formally prove the legality of the proposed \emph{\opmvotm{s}} in \secref{cmvostm}. 


\vspace{1mm}
\noindent
\textbf{Opacity:} It is a \emph{\ccs} for STMs \cite{GuerKap:Opacity:PPoPP:2008}. A sequential history $H$ is said to be \opq{} if there exists a serial history $S$ such that: (1) $S$ is equivalent to $\overline{H}$, i.e., $\evts{\overline{H}} = \evts{S}$ (2) $S$ is \legal{} and (3) $S$ respects the transactional real-time order of $H$, i.e., $\prec_H^{\tr} \subseteq \prec_S^{\tr}$. 

Finally, we show that history generated by \emph{\opmvotm{s}} satisfy correctness criteria as \opq{}.



\section{Graph Characterization of Opacity}
\label{sec:gcofo}

To prove that an STM system satisfies opacity, it is useful to consider graph characterization of histories. In this section, we describe the graph characterization of Guerraoui and Kapalka \cite{tm-book} modified for sequential histories.

Consider a history $H$ which consists of multiple version for each \tobj. The graph characterization uses the notion of \textit{version order}. Given $H$ and a \tobj{} $k$, we define a version order for $k$ as any (non-reflexive) total order on all the versions of $k$ ever created by committed transactions in $H$. It must be noted that the version order may or may not be the same as the actual order in which the versions of $k$ are generated in $H$. A version order of $H$, denoted as $\ll_H$ is the union of the version orders of all the \tobj{s} in $H$. 

Consider the history $H3$ as shown in \figref{mvostm3} $: lu_1(k_{x, 0}, null), lu_2(k_{x, 0}, null), lu_1\\(k_{y, 0}, null), lu_3(k_{z, 0}, null), ins_1(k_{x, 1}, v_{11}), ins_3(k_{y, 3}, v_{31}), ins_2(k_{y, 2}, v_{21}), ins_1(k_{z, 1},\\ v_{12}), c_1, c_2, lu_4(k_{x, 1}, v_{11}), lu_4(k_{y, 2}, v_{21}), ins_3(k_{z, 3}, v_{32}), c_3, lu_4(k_{z, 1}, v_{12}), lu_5(k_{x, 1}, \\v_{11}), lu_6(k_{y, 2}, v_{21}), c_4, c_5, c_6$. Using the notation that a committed transaction $T_i$ writing to $k_x$ creates a version $k_{x, i}$, a possible version order for $H3$ $\ll_{H3}$ is: $\langle k_{x, 0} \ll k_{x, 1} \rangle, \langle k_{y, 0} \ll k_{y, 2} \ll k_{y, 3} \rangle, \langle k_{z, 0} \ll k_{z, 1} \ll k_{z, 3} \rangle $.
\cmnt{
Consider the history $H4: l_1(ht, k_{x, 0}, NULL) l_2(ht, k_{x, 0}, NULL) l_1(ht, k_{y, 0}, NULL) l_3(ht, k_{z, 0},\\ NULL) i_1(ht, k_{x, 1}, v_{11}) i_3(ht, k_{y, 3}, v_{31}) i_2(ht, k_{y, 2}, v_{21}) i_1(ht, k_{z, 1}, v_{12}) c_1 c_2 l_4(ht, k_{x, 1}, v_{11}) l_4(ht,\\ k_{y, 2}, v_{21}) i_3(ht, k_{z, 3}, v_{32}) c_3 l_4(ht, k_{z, 1}, v_{12}) l_5(ht, k_{x, 1}, v_{11}), l_6(ht, k_{y, 2}, v_{21}) c_4, c_5, c_6$. In our representation, we abbreviate \tins{} as $i$, \tdel{} as $d$ and \tlook{} as $l$. Using the notation that a committed transaction $T_i$ writing to $k_x$ creates a version $k_{x, i}$, a possible version order for $H4$ $\ll_{H4}$ is: $\langle k_{x, 0} \ll k_{x, 1} \rangle, \langle k_{y, 0} \ll k_{y, 2} \ll k_{y, 3} \rangle, \langle k_{z, 0} \ll k_{z, 1} \ll k_{z, 3} \rangle $. 
}
\begin{figure}[H]
	\centering
	\captionsetup{justification=centering}
	\centerline{\scalebox{0.45}{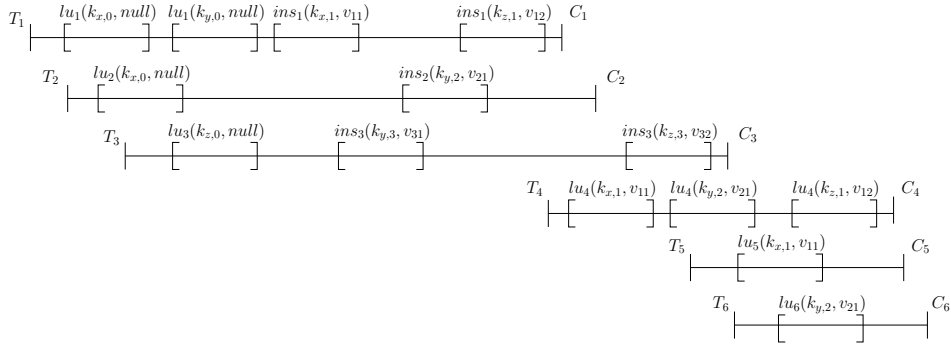}}
	\caption{History $H3$ in time line view}
	\label{fig:mvostm3}
\end{figure}
\cmnt{
\begin{figure}[tbph]
	\centerline{\scalebox{0.45}{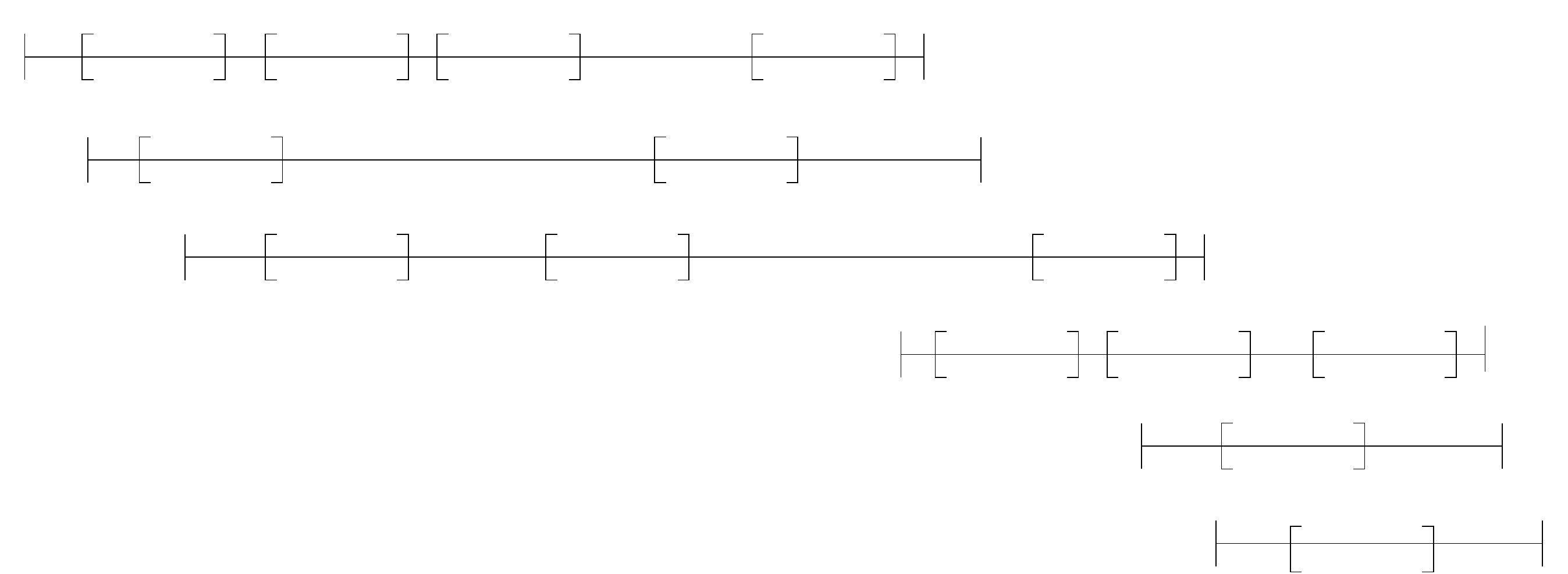}}
	\caption{History $H4$ in time line view}
	\label{fig:mvostm3}
\end{figure}
}
We define the graph characterization based on a given version order. Consider a history $H$ and a version order $\ll$. We then define a graph (called opacity graph) on $H$ using $\ll$, denoted as $\opg{H}{\ll} = (V, E)$. The vertex set $V$ consists of a vertex for each transaction $T_i$ in $\overline{H}$. The edges of the graph are of three kinds and are defined as follows:
\begin{enumerate}
\setlength\itemsep{0em}
\item \textit{\rt}(real-time) edges: If the commit of $T_i$ happens before beginning of  $T_j$ in $H$, then there exist a real-time edge from $v_i$ to $v_j$. We denote set of such edges as $\rt(H)$.
\item \textit{\rvf}(return value-from) edges: If $T_j$ invokes \rvmt on key $k_1$ from $T_i$ which has already been committed in $H$, then there exists a return value-from edge from $v_i$ to $v_j$. If $T_i$ is having \upmt{} as insert on the same key $k_1$ then $ins_i(k_{1, i}, v_{i1}) <_H c_i <_H \rvm_j(k_{1, i}, v_{i1})$. If $T_i$ is having \upmt{} as delete on the same key $k_1$ then $del_i(k_{1, i}, null) <_H c_i <_H \rvm_j(k_{1, i}, null)$. We denote set of such edges as $\rvf(H)$.
\item \textit{\mv}(multi-version) edges: This is based on version order. Consider a triplet with successful methods as  $\up_i(k_{1, i},u)$, $\rvm_j(k_{1, i},u)$, $\up_k(k_{1, k},v)$ , where $u \neq v$. As we can observe it from $\rvm_j(k_{1,i},u)$, $c_i <_H\rvm_j(k_{1,i},u)$. if $k_{1,i} \ll k_{1,k}$ then there exist a multi-version edge from $v_j$ to $v_k$. Otherwise ($k_{1,k} \ll k_{1,i}$), there exist a multi-version edge from $v_k$ to $v_i$. We denote set of such edges as $\mv(H, \ll)$.
\end{enumerate}
\cmnt{
\begin{enumerate}

\item \textit{\rt}(real-time) edges: If commit of $T_i$ happens before beginning of  $T_j$ in $H$, then there exist a real-time edge from $v_i$ to $v_j$. We denote set of such edges as $\rt(H)$.

\item \textit{\rvf}(return value-from) edges: If $T_j$ invokes \rvmt on key $k_1$ from $T_i$ which has already been committed in $H$, then there exist a return value-from edge from $v_i$ to $v_j$. If $T_i$ is having \upmt{} as insert on the same key $k_1$ then $i_i(k_{1, i}, v_{i1}) <_H c_i <_H \rvm_j(k_{1, i}, v_{i1})$. If $T_i$ is having \upmt{} as delete on the same key $k_1$ then $d_i(k_{1, i}, nil_{i1}) <_H c_i <_H \rvm_j(k_{1, i}, nil_{i1})$. We denote set of such edges as $\rvf(H)$.

\item \textit{\mv}(multi-version) edges: This is based on version order. Consider a triplet with successful methods as  $\up_i(k_{1,i},u)$ $\rvm_j(k_1,u)$ $\up_k(k_{1,k},v)$ , where $u \neq v$. As we can observe it from $\rvm_j(k_1,u)$, $c_i <_H\rvm_j(k_1,u)$. if $k_{1,i} \ll k_{1,k}$ then there exist a multi-version edge from $v_j$ to $v_k$. Otherwise ($k_{1,k} \ll k_{1,i}$), there exist a multi-version edge from $v_k$ to $v_i$. We denote set of such edges as $\mv(H, \ll)$.
\vspace{-.3cm}
\end{enumerate}
}
\noindent We now show that if a version order $\ll$ exists for a history $H$ such that it is acyclic, then $H$ is \opq. 

\begin{figure}[H]
\centerline{\scalebox{0.7}{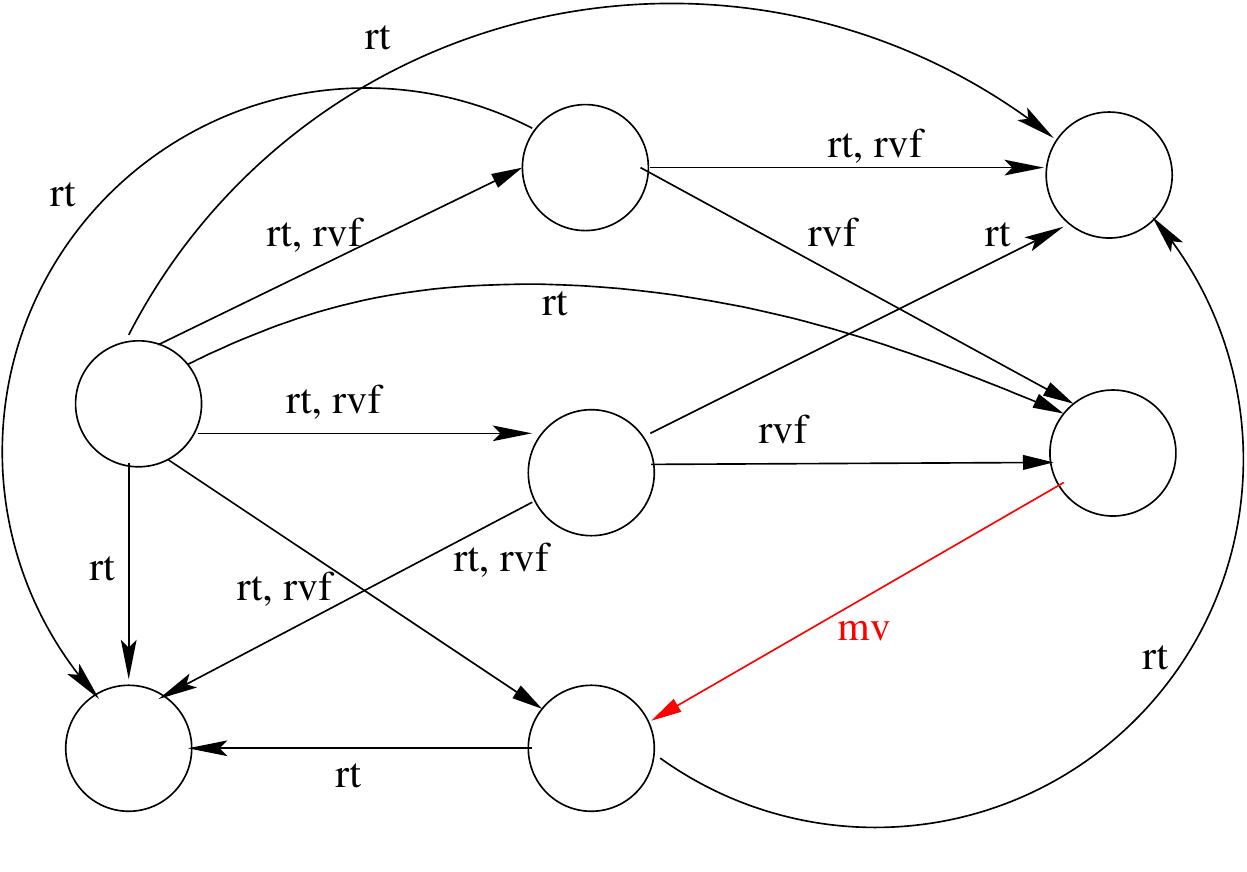}}
\caption{$\opg{H3}{\ll_{H3}}$}
\label{fig:mvostm1}
\end{figure}

Using this construction, the $\opg{H3}{\ll_{H3}}$ for history $H3$ and $\ll_{H3}$ is given above is shown in \figref{mvostm1}. The edges are annotated. The only \mv{} edge from $T_4$ to $T_3$ is because of \tobj{s} $k_y, k_z$. $T_4$ lookups value $v_{12}$ for $k_z$ from $T_1$ whereas $T_3$ also inserts $v_{32}$ to $k_z$ and commits before $lu_4(k_{z,1}, v_{12})$.

Given a history $H$ and a version order $\ll$, consider the graph $\opg{\overline{H}}{\ll}$. While considering the $\rt{}$ edges in this graph, we only consider the real-time relation of $H$ and not $\overline{H}$. It can be seen that $\prec_H^{RT} \subseteq \prec_{\overline{H}}^{RT}$ but with this assumption, $\rt(H) = \rt(\overline{H})$. Hence,  we get the following property, 

\begin{property}
\label{prop:hoverh}
The graphs $\opg{H}{\ll}$ and $\opg{\overline{H}}{\ll}$ are the same for any history $H$ and $\ll$. 
\end{property}
\begin{definition}
\label{def:seqver}
For a \tseq{} history $S$, we define a version order $\ll_S$ as follows: For two version $k_{x,i}, k_{x,j}$ created by committed transactions $T_i, T_j$ in $S$, $\langle k_{x,i} \ll_S k_{x,j} \Leftrightarrow T_i <_S T_j \rangle $. 
\end{definition}
Now we show the correctness of our graph characterization using the following lemmas and theorem. 

\begin{lemma}
\label{lem:seracycle}
Consider a \legal{} \tseq{} history $S$. Then the graph $\opg{S, \ll_S}$ is acyclic.
\end{lemma}

\begin{proof}
We numerically order all the transactions in $S$ by their real-time order by using a function \textit{\ordfn}. For two transactions $T_i, T_j$, we define $\ord{T_i} < \ord{T_j} \Leftrightarrow T_i <_S T_j$. Let us analyze the edges of $\opg{S, \ll_S}$ one by one: 
\begin{itemize}
\item \rt{} edges: It can be seen that all the \rt{} edges go from a lower \ordfn{} transaction to a higher \ordfn{} transaction. 

\item \rvf{} edges: If $T_j$ lookups $k_x$ from $T_i$ in $S$ then $T_i$ is a committed transaction with $\ord{T_i} < \ord{T_j}$. Thus, all the \rvf{} edges from a lower \ordfn{} transaction to a higher \ordfn{} transaction.

\item \mv{} edges: Consider a successful \rvmt{} $\rvm_j(k_x, u)$ and a committed transaction $T_k$ writing $v$ to $k_x$ where $u \neq v$. Let $c_i$ be $\rvm_j(k_x, u)$'s \lastw. Thus, $\up_i(k_{x,i}, u) \in \evts{T_i}$. Thus, we have that $\ord{T_i} < \ord{T_j}$. Now there are two cases w.r.t $T_i$: (1) Suppose $\ord{T_k} < \ord{T_i}$. We now have that $T_k \ll T_i$. In this case, the mv edge is from $T_k$ to $T_i$. (2) Suppose $\ord{T_i} < \ord{T_k}$ which implies that $T_i \ll T_k$. Since $S$ is legal, we get that $\ord{T_j} < \ord{T_k}$. This case also implies that there is an edge from $\ord{T_j}$ to $\ord{T_k}$. Hence, in this case as well the \mv{} edges go from a transaction with lower \ordfn{} to a transaction with higher \ordfn{}. 

\end{itemize}

Thus, in all the three cases the edges go from a lower \ordfn{} transaction to higher \ordfn{} transaction. This implies that the graph is acyclic. 
\end{proof}

\begin{lemma}
\label{lem:eqv_hist_mvorder}
Consider two histories $H, H'$ that are equivalent to each other. Consider a version order $\ll_H$ on the \tobj{s} created by $H$. The mv edges $\mv(H, \ll_H)$ induced by $\ll_H$ are the same in $H$ and $H'$.
\end{lemma}

\begin{proof}
Since the histories are equivalent to each other, the version order $\ll_H$ is applicable to both of them. It can be seen that the \mv{} edges depend only on events of the history and version order $\ll$. It does not depend on the ordering of the events in $H$. Hence, the \mv{} edges of $H$ and $H'$ are equivalent to each other. 
\end{proof}

\noindent Using these lemmas, we prove the following theorem.

\begin{theorem}
\label{thm:opg}
A \valid{} history H is opaque iff there exists a version order $\ll_H$ such that $\opg{H}{\ll_H}$ is acyclic.
\end{theorem}

\begin{proof}
\textbf{(if part):} Here we have a version order $\ll_H$ such that $G_H=\opg{H}{\ll}$ is acyclic. Now we have to show that $H$ is opaque. Since the $G_H$ is acyclic, a topological sort can be obtained on all the vertices of $G_H$. Using the topological sort, we can generate a \tseq{} history $S$. It can be seen that $S$ is equivalent to $\overline{H}$. Since $S$ is obtained by a topological sort on $G_H$ which maintains the real-time edges of $H$, it can be seen that $S$ respects the \rt{} order of $H$, i.e $\prec_H^{RT} \subseteq \prec_S^{RT}$. 

Similarly, since $G_H$ maintains return value-from (\rvf{}) order of $H$, it can be seen that if $T_j$ lookups $k_x$ from $T_i$ in $H$ then $T_i$ terminates before $lu_j(k_x)$ and $T_j$ in $S$. Thus, $S$ is \valid. Now it remains to be shown that $S$ is \legal. We prove this using contradiction. Assume that $S$ is not legal. Thus, there is a successful \rvmt{} $\rvm_j(k_x, u)$ such that its \lastw{} in $S$ is $c_k$ and $T_k$ updates value $v (\neq u)$ to $k_x$, i.e $\up_k(k_{x,k}, v) \in \evts{T_k}$. Further, we also have that there is a transaction $T_i$ that inserts $u$ to $k_x$, i.e $\up_i(k_{x,i}, u) \in \evts{T_i}$. Since $S$ is \valid, as shown above, we have that $T_i \prec_{S}^{RT} T_k \prec_{S}^{RT} T_j$.

Now in $\ll_H$, if $k_{x,k} \ll_H k_{x,i}$ then there is an edge from $T_k$ to $T_i$ in $G_H$. Otherwise ($k_{x,i} \ll_H k_{x,k}$), there is an edge from $T_j$ to $T_k$. Thus, in either case, $T_k$ can not be in between $T_i$ and $T_j$ in $S$ contradicting our assumption. This shows that $S$ is legal.



\textbf{(Only if part):} Here we are given that $H$ is opaque and we have to show that there exists a version order $\ll$ such that $G_H=\opg{H}{\ll} (=\opg{\overline{H}}{\ll}$, \propref{hoverh}) is acyclic. Since $H$ is opaque there exists a \legal{} \tseq{} history $S$ equivalent to $\overline{H}$ such that it respects real-time order of $H$. Now, we define a version order for $S$, $\ll_S$ as in \defref{seqver}. Since the $S$ is equivalent to $\overline{H}$, $\ll_S$ is applicable to $\overline{H}$ as well. From \lemref{seracycle}, we get that $G_S=\opg{S}{\ll_S}$ is acyclic. Now consider $G_H = \opg{\overline{H}}{\ll_S}$. The vertices of $G_H$ are the same as $G_S$. Coming to the edges, 

\begin{itemize}
\item \rt{} edges: We have that $S$ respects real-time order of $H$, i.e $\prec_{H}^{RT} \subseteq \prec_{S}^{RT}$. Hence, all the \rt{} edges of $H$ are a subset of $S$. 

\item \rvf{} edges: Since $\overline{H}$ and $S$ are equivalent, the return value-from relation of $\overline{H}$ and $S$ are the same. Hence, the \rvf{} edges are the same in $G_H$ and $G_S$. 

\item \mv{} edges: Since the version-order and the \op{s} of the $H$ and $S$ are the same, from \lemref{eqv_hist_mvorder} it can be seen that $\overline{H}$ and $S$ have the same \mv{} edges as well.
\end{itemize}

Thus, the graph $G_H$ is a subgraph of $G_S$. Since we already know that $G_S$ is acyclic from \lemref{seracycle}, we get that $G_H$ is also acyclic. 
\end{proof}

\section{\emph{\opmvotm{s}} Design and Data Structure}
\label{sec:mvdesign}
This section describes the design and data structure of optimized \emph{\mvotm{s}} (or \emph{\opmvotm{s}}). Here, we propose hash-table and list based \emph{\opmvotm{s}} as \emph{\ophmvotm} and \emph{\oplmvotm} respectively. \emph{\opmvotm{s}} are generic for other data structure as well. \emph{\ophmvotm} is a \tab based \opmvotm that explores the idea of multiple versions in \otm{s} for \tab object to achieve greater concurrency. The design of \ophmvotm is similar to \hmvotm{} \cite{Juyal+:MVOSTM:SSS:2018} consisting of $B$ buckets. All the keys of the \tab in the range $\mathcal{K}$ are statically allocated to one of these buckets. 

Each bucket consists of linked-list of nodes along with two sentinel nodes \emph{head} and \emph{tail} with values -$\infty$ and +$\infty$ respectively. The structure of each node is as $\langle key, ~ lock, ~ \\ marked, ~ vl, ~ nnext \rangle$. The $key$ is a unique value from the set of all keys $\mathcal{K}$. All the nodes are stored in increasing order in each bucket as shown in \figref{1mvostmdesign} (a), similar to any linked-list based concurrent set implementation \cite{Heller+:LazyList:PPL:2007, Harris:NBList:DISC:2001}. In the rest of the document, we use the terms key and node interchangeably. To perform any operation on a key, the corresponding $lock$ is acquired. $marked$ is a boolean field which represents whether the key is deleted or not. The deletion is performed in a lazy manner similar to the concurrent linked-lists structure \cite{Heller+:LazyList:PPL:2007}. If the $marked$ field is true then key corresponding to the node has been logically deleted; otherwise, it is present. The $vl$ field of the node points to the version list (shown in \figref{1mvostmdesign} (b)) which stores multiple versions corresponding to the key. The last field of the node is $nnext$ which stores the address of the next node. It can be seen that the list of keys in a bucket is as an extension of \emph{\lazy} \cite{Heller+:LazyList:PPL:2007}. Given a node $n$ in the linked-list of bucket $B$ with key $k$, we denote its fields as $n.key$ (or $k.key$), $~ n.lock$ (or $k.lock$), $~ n.marked$ (or $k.marked$), $~ n.vl$ (or $k.vl$), $~ n.nnext$ (or $k.nnext$).



\cmnt{
\begin{figure}
	\centering
	\captionsetup{justification=centering}
	\centerline{\scalebox{0.47}{\input{figs/mvostm10.pdf_t}}}
	\caption{\htmvotm design}
	\label{fig:1mvostmdesign1}
\end{figure}
}
\begin{figure}
	\centering
	\captionsetup{justification=centering}
	\centerline{\scalebox{0.4}{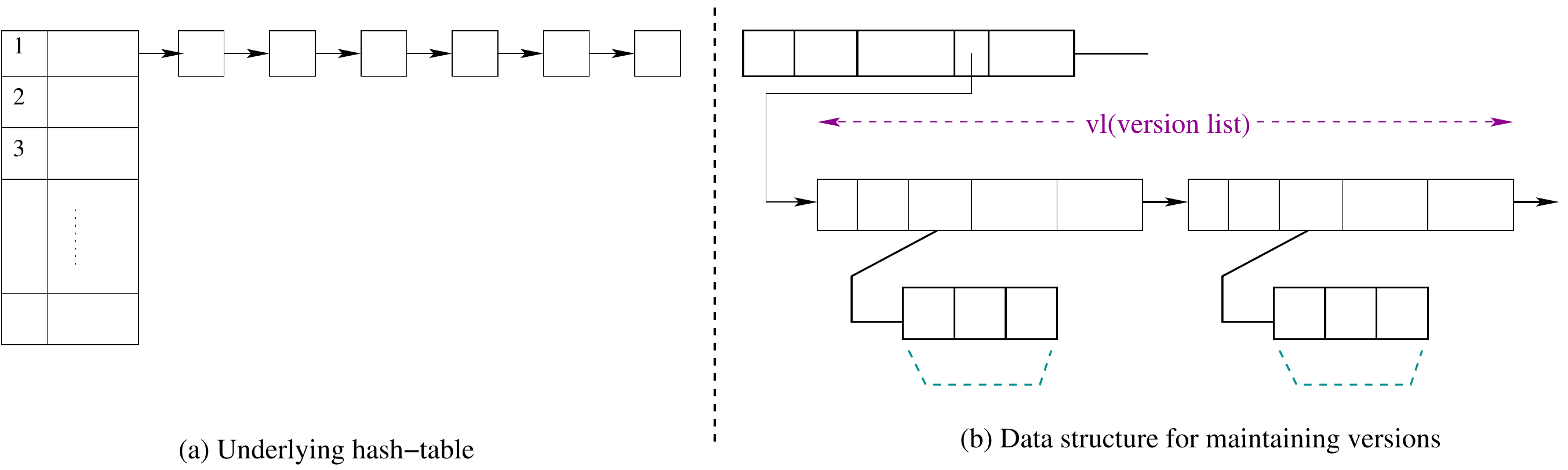}}
	\caption{Optimized \emph{HT-MVOSTM} design}
	\label{fig:1mvostmdesign}
\end{figure}

The structure of each version in the $vl$ of a key $k$ is $\langle ts, ~ val, ~ rvl, ~ max_{rvl}, ~ vnext \rangle$ as shown in \figref{1mvostmdesign} (b). The field $ts$ denotes the unique timestamp of the version. In our algorithm, every transaction is assigned a unique timestamp when it begins which is also its $id$. Thus $ts$ of this version is the timestamp of the transaction that created it. All the versions in the $vl$ of $k$ are sorted by $ts$. Since the timestamps are unique, we denote a version, $ver$ of a node $n$ with key $k$ having $ts$ $j$ as $n.vl[j].ver$ or $k.vl[j].ver$. The corresponding fields in the version as $k.vl[j].ts, ~ k.vl[j].val, ~ k.vl[j].rvl, ~ k.vl[j].max_{rvl}, \\~ k.vl[j].vnext$. 

The field $val$ contains the value updated by an update transaction. If this version is created by an insert \mth $\tins_i(ht, k, v)$ by transaction $T_i$, then $val$ will be $v$. On the other hand, if the \mth is $\tdel_i(ht, k, v)$ then $val$ will be $null$. In this case, as per the algorithm, the node of key $k$ will also be marked. \ophmvotm algorithm does not immediately physically remove deleted keys from the \tab. The need for this is explained below. Thus an \rvmt (\tdel or \tlook) on key $k$ can return $null$ when it does not find the key or encounters a $null$ value for $k$. 

The $rvl$ field stands for \emph{return value list} which is a list of all the transactions that executed \rvmt{} on this version, i.e., those transactions which returned $val$. The first optimization in \emph{\ophmvotm} to reduce the traversal time of $rvl$, we have used $max_{rvl}$  which contains the maximum $ts$ of the transaction that executed \rvmt{} on this version. The field $vnext$ points to the next available version of that key. 

In order to increase the efficiency and utilize the memory properly, We propose two variants of \emph{\ophmvotm} as follows: First, we apply garbage collection (or GC) on the versions and propose \emph{\ophmvotmgc}. It maintains unbounded versions in $vl$ (the length of the list) while deleting the unwanted versions using garbage collection scheme. Second, we propose \emph{\ophkotm} which maintains the bounded number of versions such as $K$ and improves the efficiency further. Whenever a new version $ver$ is created and is about to be added to $vl$, the length of $vl$ is checked. If the length becomes greater than $K$, the version with lowest $ts$ (i.e., the oldest) is replaced with the new version $ver$ and thus maintaining the length back to $K$. 

We propose \emph{\oplmvotm{s}} while considering the bucket size as 1 in \emph{\ophmvotm}. Along with this, we propose two variants of \emph{\oplmvotm}
as \emph{\oplmvotmgc} and \emph{\oplkotm} which applies the garbage collection scheme in unbounded versions and bounded $K$ versions for list based object respectively  similar to \emph{\ophmvotm}.
\ignore{

	\begin{figure}
	\captionsetup{justification=centering}
		\centerline{\scalebox{0.40}{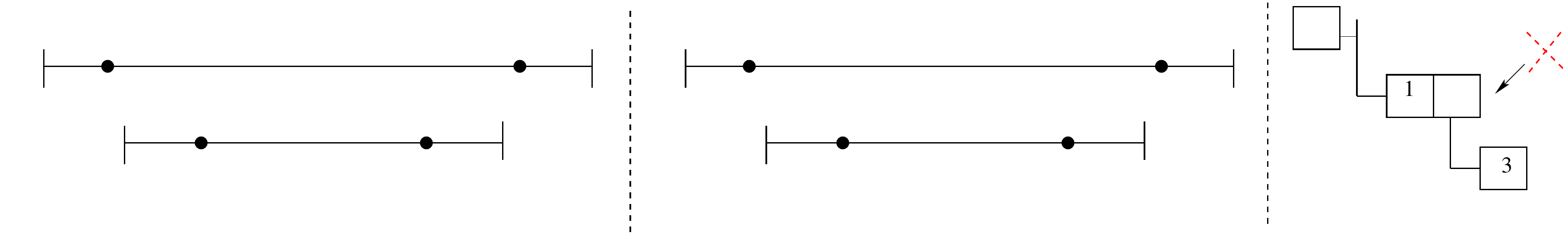}}
		\caption{Need of maintaining deleted node in underlying DS to satisfy opacity}
		\label{fig:mvostm81}
	\end{figure}	

\textbf{Why do we need to store the deleted node?} This will be clear by the \figref{mvostm81}, where we have two concurrent transactions $T_2$ and $T_3$. History in the \figref{mvostm81} (a) is not opaque because we can't come up with any serial order. To make it serial (or opaque) the second method $ins_2(ht, k_{3,2})$ of transaction $T_2$ has to return abort. For that \hmvotm scheduler have to keep the information about the already executed conflicting method $del_3(ht, k_{3,0}, NULL)$ of transaction $T_3$ using $ts$. Therefore, $del_3(ht, k_{3,0}, NULL)$ of transaction $T_3$ add itself into $T_1.rvl$ (assume that transaction $T_1$ has already inserted a version on key $k_3$ refer \figref{mvostm81} (c). So in future if any lower time-stamp transaction less than $T_3$ will come then that lower transaction will ABORT (in this case transaction $T_2$ is aborting in (\figref{mvostm81} (b)) because higher time-stamp already present in the $rvl$ (\figref{mvostm81} (c)) of the same version. After aborting $T_2$ we will get the equivalent serial history $T_2$ followed by $T_3$ 
}

\vspace{1mm}
\noindent
\textbf{Marked Version Nodes:} \ophmvotm stores keys even after they have been deleted (the version of the nodes which have $marked$ field as true). This is because some other concurrent transactions could read from a different version of this key and not the $null$ value inserted by the deleting transaction. Consider for instance the transaction $T_1$ performing $lu_1(ht, k_2, v_0)$ as shown in \figref{pop} (b). Due to the presence of previous version $v_0$, \ophmvotm returns this earlier version $v_0$ for $lu_1(ht, k_2, v_0)$ \mth. Whereas, it is not possible for \hotm to return the version $v_0$  because $k_1$ has been removed from the system by delete method of higher timestamp transaction $T_2$ than $T_1$. In that case, $T_1$ would have to be aborted. Thus as explained in \secref{intro}, storing multiple versions increases the concurrency. 

To store deleted keys along with the live keys (or unmarked node) in a \lazy will increase the traversal time to access unmarked nodes. Consider \figref{nostm2}, in which there are four keys $\langle k_2, k_4, k_8, k_{11}\rangle$ present in the list. Here $\langle k_2, k_4, k_8 \rangle$ are marked (or deleted) nodes while $k_{11}$ is unmarked. Now, consider accessing the key $k_{11}$ by \ophmvotm as a part of one of its methods. Then \ophmvotm would have to unnecessarily traverse the marked nodes to reach key $k_{11}$. 

\begin{figure}
	\centering
	\begin{minipage}[b]{0.49\textwidth}
		\centering
		\captionsetup{justification=centering}
		\scalebox{.38}{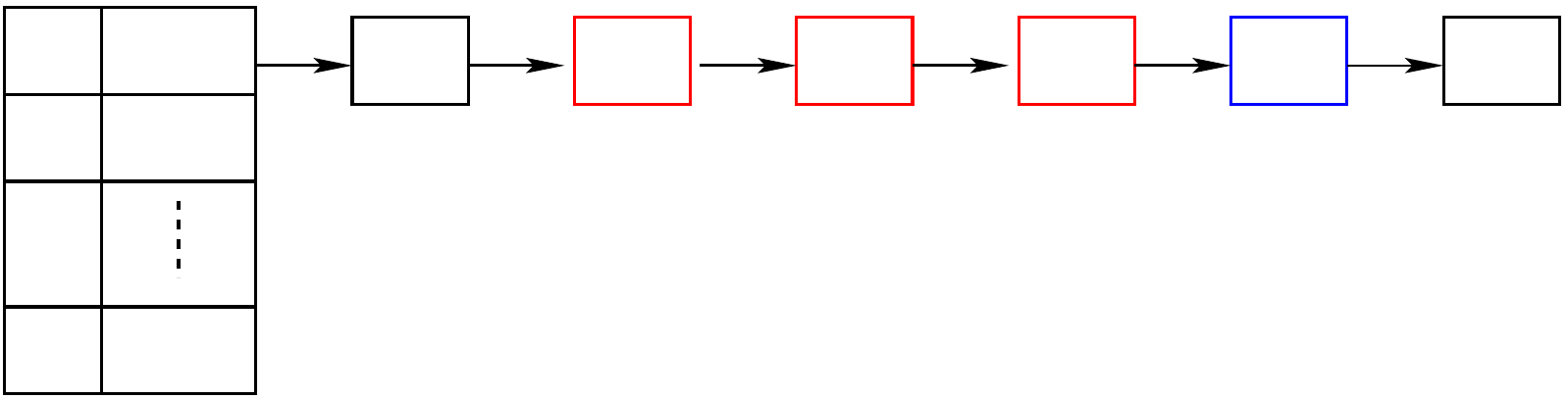}
		\caption{Searching $k_{11}$ over \emph{lazy-list}}
		\label{fig:nostm2}
	\end{minipage}   
	\hfill
	\begin{minipage}[b]{0.49\textwidth}
		\scalebox{.38}{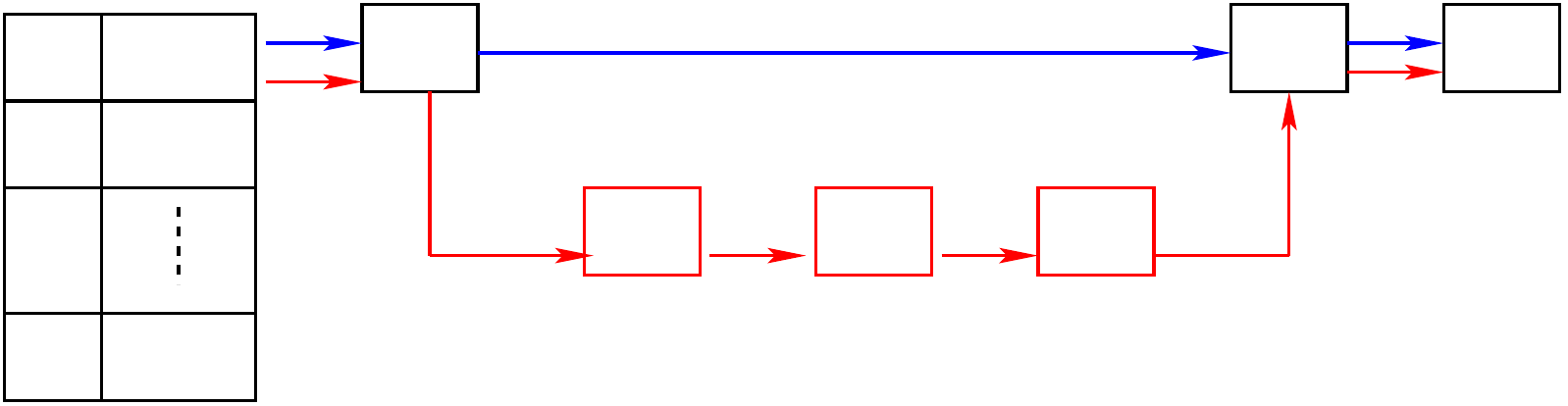}
		\centering
		\captionsetup{justification=centering}
		\caption{Searching $k_{11}$ over $\lsl{}$}
		\label{fig:nostm}
	\end{minipage}
\end{figure}
This motivated us to modify the \lazy structure of nodes in each bucket to form a skip list based on red and blue links. We denote it as \emph{red-blue lazy-list} or \emph{\lsl}. This idea was earlier explored by Peri et al. in developing \otm{s} \cite{Peri+:OSTM:Netys:2018}. $\lsl{}$ consists of nodes with two links, red link (or \rn) and blue link (or \bn). The node which is not marked (or not deleted) are accessible from the head via \bn{}. While all the nodes including the marked ones can be accessed from the head via \rn. With this modification, let us consider the above example of accessing unmarked key $k_{11}$. It can be seen that $k_{11}$ can be accessed much more quickly through \bn as shown in \figref{nostm}. Using the idea of $\lsl{}$, we have modified the structure of each node as \emph{$\langle$ key, lock, marked, vl, \rn, \bn $\rangle$}. Further, for a bucket $B$, we denote its linked-list as $B.\lsl$. 

\ignore{

Now, if node corresponding to the key doesn't exist in the underlying DS then \textbf{How we will maintain the node time-stamp by \rvmt{}?} This case will occur when node corresponding to the key is not present in \bn{} as well as \rn{}. Then \rvmt{} will create the node corresponding to the key in \rn{} with marked field as true and append the $0^{th}$ version in its $vl$ and add itself into $0^{th}.rvl$. Inserting the $0^{th}$ version ensures that the transactions which contain only \rvmt{s} will never abort. 
\begin{figure}
	\centering
	\captionsetup{justification=centering}
	\scalebox{.4}{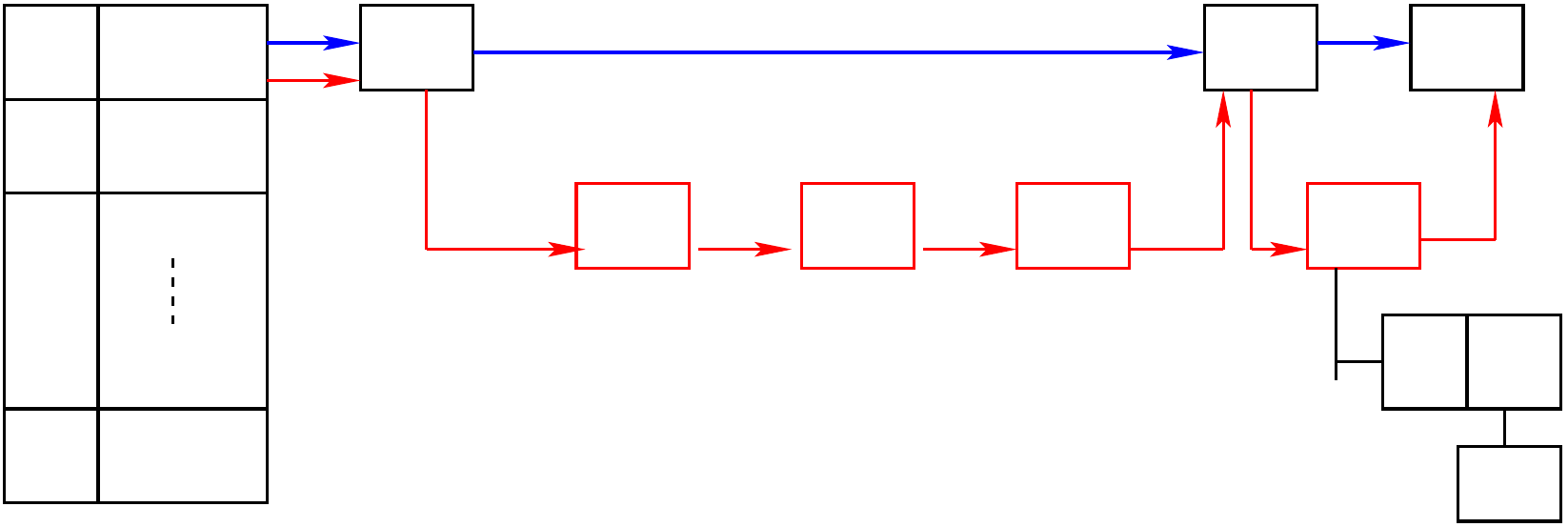}
	\caption{Execution under \lsl{}. $k_{20}$ is added in \lsl{} if not present.}
	\label{fig:nostm1}
\end{figure}

Consider the \figref{nostm}, in which \npluk{} want to search a key $k_{20}$. Here, the key $k_{20}$ is not present in the \bn{} as well as \rn{}. So it will create the node corresponding to the key in \rn{} with marked field as true and append the $0^{th}$ version in its $vl$ and add itself into $0^{th}.rvl$ refer \figref{nostm1}. 

Each transaction maintains local log in the form of $L\_txlog$ where they store unique transaction id, status and $L\_list$. If transaction is currently executing or committed or aborted due to some inconsistency then status will be live, commit or abort, respectively. $L\_list$ is a vector which contains $\langle bucket\_id, key, val, preds, currs, op\_status, opn \rangle$. The $op\_status$ is an operation status which could be $OK$ or $FAIL$.
Each \mth{} identifies the location corresponding to the key in \bn{} and \rn{} from \lsl{} where they have to work and store the $preds$ and $currs$ in the form of an array. $opn$ stands for operation name which could be \npluk, \npins{}, \npdel{} which will work on $key$. 

}

\section{Working of \ophmvotm{}} 
\label{sec:pcode}


\ophmvotm exports \tbeg, \tins, \tdel, \tlook, and \tryc \mth{s} as explained in \secref{model}. Among them \tdel, \tlook are return-value methods (or \rvmt{s}) while \tins, \tdel are update methods (or \upmt{s}). We treat \tdel as both \rvmt as well as \upmt. The \rvmt{s} return the current value of the key. 
The \upmt{s}, update to the keys are first noted down in the local log, \emph{\llog}. Then in the \tryc \mth after successful validations of these updates are transferred to the shared memory.
\ignore{

\begin{figure}
	\centering	
	\captionsetup{justification=centering}
	\centerline{\scalebox{0.5}{\input{figs/pnetys8.pdf_t}}}
	\caption{Transaction life cycle of \mvotm}	
	\label{fig:nostm24}
\end{figure}
}
We now explain the working of each method as follows: 


\vspace{1mm}
\noindent 
\textbf{\emph{t\_begin()}:} A thread invokes a new transaction $T_i$ using this method. The transaction $T_i$ local log $\llog_i$ is initialized at \Lineref{bg2}. This \mth returns a unique id to the invoking thread by incrementing an atomic counter at \Lineref{bg3}. This unique id is also the timestamp of the transaction $T_i$. For convenience, we use the notation that $i$ is the timestamp (or id) of the transaction $T_i$. 

\begin{algorithm}
	\scriptsize
	\caption{\emph{t\_begin()}: It provides the local log and unique id to each transaction.}
	\setlength{\multicolsep}{0pt}
	\begin{algorithmic}[1]
		\makeatletter\setcounter{ALG@line}{0}\makeatother
		\Procedure{\emph{t\_begin{()}}}{}\label{lin:bg1}
		\State txLog $\gets$ new txLog(). \Comment{Initialize the local log of transaction} \label{lin:bg2}
		\State $t\_id$ $\gets$ get\&inc(\emph{counter}). \Comment{Get the unique transaction id ($t\_id$) while incrementing the \emph{counter} atomically} \label{lin:bg3}
		\State return $t\_id$.
		\EndProcedure
	\end{algorithmic}
\end{algorithm}

\vspace{1mm}
\noindent 
\textbf{\rvmt{s}}: It can be either $\tdel(ht, k, v)$ or $\tlook(ht, k, v)$. Both these \mth{s} return the current value of key $k$. \algoref{rvmt} gives the high level overview of these \mth{s}. First, the algorithm checks to see if the given key is already in the local log, $\llog_i$ of $T_i$ (\linref{rvm-chk_log}). If the key is already there then the current \rvmt is not the first method on $k$ and is a subsequent method of $T_i$ on $k$. So, we can return the value of $k$ from the $\llog_i$. 

If the key is not present in the $\llog_i$, then \ophmvotm searches into shared memory. Specifically, it searches the bucket to which $k$ belongs to. Every key in the range $\mathcal{K}$ is statically allocated to one of the $B$ buckets. So the algorithms search for $k$ in the corresponding bucket, say $B_k$ to identify the appropriate location, i.e., identify the correct \emph{predecessor} or $pred$ and \emph{current} or  $curr$ keys in the \lsl of $B_k$ without acquiring any locks similar to the search in \lazy \cite{Heller+:LazyList:PPL:2007}. Since each key has two links, \rn and \bn, the algorithm identifies four node references: two $pred$ and two $curr$ according to red and blue links. They are stored in the form of an array with $\bp$ and $\bc$ corresponding to blue links; $\rp$ and $\rc$ corresponding to red links. If both $\rp$ and $\rc$ nodes are unmarked then the $pred, curr$ nodes of both red and blue links will be the same, i.e., $\bp = \rp$ and $\rc = \bc$. Thus depending on the marking of $pred, curr$ nodes, a total of two, three or four different nodes will be identified. Here, the search ensures that $\bp.key \leq \rp.key < k \leq \rc.key \leq \bc.key$. 

Next, the re-entrant locks on all the $pred, curr$ keys are acquired in increasing order to avoid the deadlock. Then all the $pred$ and $curr$ keys are validated by \emph{rv\_Validation()} in \linref{rvm-chk_valid} as follows: (1) If $pred$ and $curr$ nodes of blue links are not marked, i.e, $(\neg \bp.marked) ~ \&\& ~ (\neg \bc.marked)$. (2) If the next links of both blue and red $pred$ nodes point to the correct $curr$ nodes: $(\bp.\bn = \bc) ~ \&\& ~ (\rp.\rn = \rc)$ at \Lineref{rvv1}. 

If any of these checks fail, then the algorithm retries to find the correct $pred$ and $curr$ keys. It can be seen that the validation check is similar to the validation in concurrent \lazy \cite{Heller+:LazyList:PPL:2007}.

Next, we check if $k$ is in $B_k.\lsl$. If $k$ is not in $B_k$, then we create a new node $n$ for $k$ as: $\langle key=k, lock=false, marked = true, vl=ver, nnext=\phi \rangle$ and insert it into $B_k.\lsl$ such that it is accessible only via \rn. 
This node will have a single version $ver$ as $\langle ts=0, val=null, rvl=i, max_{rvl} = i, vnext=\phi \rangle$. Here invoking transaction $T_i$ is creating a version with timestamp $0$ to ensure that \rvmt{s} of other transactions will never abort. As we have explained in \figref{pop} (b) of \secref{intro}, even after $T_2$ deletes $k_2$, the previous value of $v_0$ is still retained. Thus, when $T_1$ invokes $lu$ on $k_2$ after the delete on $k_2$ by $T_2$, \ophmvotm will return $v_0$ (as previous value). Hence, each \rvmt will find a version to read while maintaining the infinite version corresponding to each key $k$. $marked$ field sets to true because it access by $\rn$ only. In $rvl$ and $max_{rvl}$, $T_i$ adds the timestamp as $i$ in it and $vnext$ is initialized to empty value. Since $val$ is null and the $n$, this version and the node are not technically inserted into $B_k.\lsl$. 

If $k$ is in $B_k.\lsl$ then, $k$ is the same as $\rc$ or $\bc$ or both. Let $n$ be the node of $k$ in $B_k.\lsl$. We then find the version of $n$, $ver_j$ which has the timestamp $j$ such that $j$ has the largest timestamp smaller than $i$ (timestamp of $T_i$). Add $i$ to $ver_j$'s $rvl$ (\linref{rvm-add_i}). $max_{rvl}$ maintains the maximum timestamp among all rv\_methods read from this version at \Lineref{rvm-max}. Then release the locks, update the local log $\llog_i$ in \linref{rvm-unlock} and return the value stored in $ver_j.val$ in \linref{rvm-ret}. 
\begin{algorithm}
	\scriptsize
	\caption{\emph{\rvmt:} It can be either $\tdel_i(ht, k, v)$ or $\tlook_i(ht, k, v)$ on key $k$ that maps to bucket $B_k$ of hash-table $ht$.} \label{algo:rvmt} 	
	\begin{algorithmic}[1]
		\makeatletter\setcounter{ALG@line}{5}\makeatother
		\Procedure{$\rvmt_i$}{$ht, k, v$}		
		\If{($k \in \llog_i$)} \label{lin:rvm-chk_log}
			\State Update the local log and return $val$. \label{lin:rvm-pres_log}
		\Else \label{lin:rvm-els_log}
			\State Search in \lsl to identify the $preds[]$ and $currs[]$ for $k$ using \bn and \rn in bucket $B_k$. \label{lin:rvm-search}
			\State Acquire the locks on $preds[]$ and $currs[]$ in increasing order. \label{lin:rvm-locks}
			\If{($!rv\_Validation()$)} \label{lin:rvm-chk_valid}
				\State Release the locks and goto \linref{rvm-search}. \label{lin:rvm-retry}
			\EndIf 
			\If{($k  ~ \notin ~ B_k.\lsl$)} \label{lin:rvm-chk_k} 
				\State Create a new node $n$ with key $k$ as: $\langle$ \emph{key = k, lock = false, marked = true, vl = ver, nnext = }$\phi \rangle$. \label{lin:rvm-chk_k1} 
				\State /*The $vl$ consists of a single element $ver$ with $ts$ as 0*/
				\State Create the version $ver$ as: $\langle \emph{$ts=0, val=null, rvl=i, max_{rvl}=i, vnext=\phi$}\rangle$.  
				\State Insert $n$ into $B_k.\lsl$ such that it is accessible only via \rn{s}. \Comment{$n$ is marked} \label{lin:rvm-ver_abs}  	\State Release the locks; update the $\llog_i$ with $k$.
				\State return $null$.
			\EndIf 
			\State Identify the version $ver_j$ with $ts=j$ such that $j$ is the largest timestamp smaller than $i$.
			\State Add $i$ into the $rvl$ of $ver_j$. \label{lin:rvm-add_i} 
			\If{($ver_j.max_{rvl}$ $<$ $i$)}
			\State Set $ver_j.max_{rvl}$ to $i$. \label{lin:rvm-max} 
			\EndIf
			\State $retVal = ver_j.val$.
			\State Release the locks; update the $\llog_i$ with $k$ and $retVal$. \label{lin:rvm-unlock} 
		\EndIf 
		\State return $retVal$. \label{lin:rvm-ret}
		\EndProcedure
	\end{algorithmic}
\end{algorithm}

\ignore {
The \emph{rv\_Validation()} is called by the \rvmt{} and \upmt{}. It will identify the conflicts among the concurrent methods of different transactions. It will be more clear by the \figref{mvostm9}, where two concurrent conflicting methods of different transactions are working on the same key $k_3$. Initially, at stage $s_1$ in \figref{mvostm9} (c) both the conflicting methods $ins_1(k_3)$ and $lu_2(k_3)$ optimistically (without acquiring locks) identify the same $preds$ and $currs$ for key $k_3$ from underlying $CDS$ in \figref{mvostm9} (a). At stage $s_2$ in \figref{mvostm9} (c), method $ins_1(k_3)$ of transaction $T_1$ acquired the lock on $preds$ and $currs$ and inserted the node into it (\figref{mvostm9} (b)). After successful insertion by $T_1$, $preds$ and $currs$ will change for $lu_2(k_3)$ at stage $s_3$ in \figref{mvostm9} (c). It will caught via \emph{rv\_Validation} method at \Lineref{rvv1} when $(\bp.\bn \neq \bc)$ for $lu_2(k_3)$. After that again it will find the new $preds$ and $currs$ for $lu_2(k_3)$ and eventually it will commit.

\begin{algorithm}
	\scriptsize
	\caption{\emph{rv\_Validation()} }
	\setlength{\multicolsep}{0pt}
	\begin{algorithmic}[1]
		\makeatletter\setcounter{ALG@line}{14}\makeatother
		\Procedure{rv\_validation{()}}{}
		\If{$((\bp.marked) || (\bc.marked) ||(\bp.\bn) \neq \bc || (\rp.\rn) \neq {\rc})$}\label{lin:rvv1}
		\State return $false$. \label{lin:rvv2}
		\Else \label{lin:rvv3}
		\State return $true$. \label{lin:rvv4}
		\EndIf 
		\EndProcedure
	\end{algorithmic}
\end{algorithm}

\cmnt{
\begin{figure}[tbph]
\captionsetup{justification=centering}
	\centerline{\scalebox{0.47}{\input{figs/mvostm9.pdf_t}}}
	\caption{Method validation}
	\label{fig:mvostm91}
\end{figure}
}

\begin{figure}[tbph]
\captionsetup{justification=centering}
	\centerline{\scalebox{0.42}{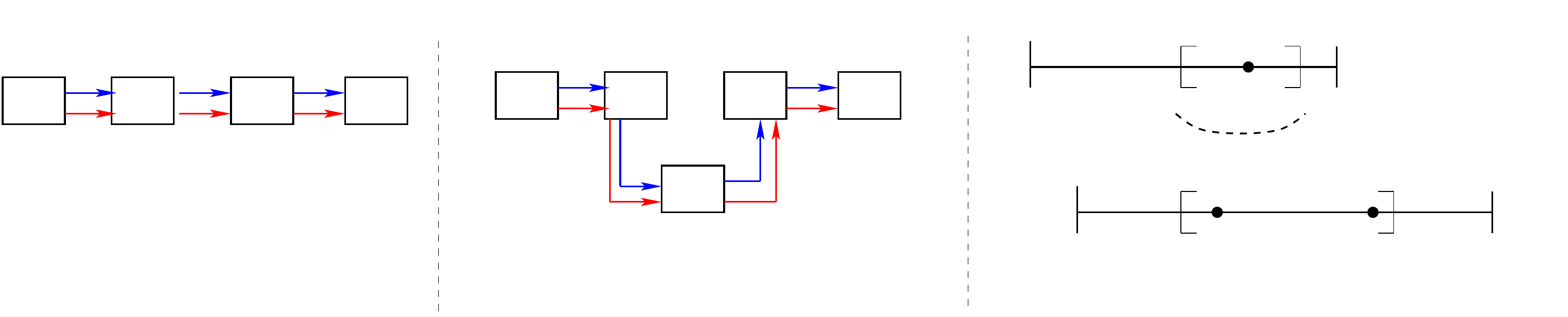}}
	\caption{Interference validation (or rv\_Validation)}
	\label{fig:mvostm9}
\end{figure}
}	

\ignore {

We are considering the chaining \tab{} as a underlying data structure where chaining is done via lazy-list refer \figref{1mvostmdesign} a). Each bucket of the \tab{} is having two sentinel nodes: $head$ and $tail$. $Head$ and $tail$ are initialized as -$\infty$ and +$\infty$ respectively. Keys $\langle k_1, k_2, ...k_n \rangle$ are added in increasing order in the list between the sentinel nodes. 

Each key is maintaining the multiple versions in increasing order of timestamp (\figref{1mvostmdesign} b)). 
For each key $k_1$ of transaction $T_i$, we maintain $k_1.vl$ (version list) which is a list consisting of version tuples in the form $\langle ts, val, mark, rvl, vnext \rangle$.  Description of each field as: $ts$ stands for timestamp which is unique for each transaction, $val$ is the value written by any transaction corresponding to the key, $mark$ is the Boolean variable which can be true or false (if the method corresponding to the key is \npdel{} then the value of $mark$ field will be true, ($T$) and if the method corresponding to the key is \npins{} then the value of $mark$ field will be false, ($F$)), $rvl$ represents $return$-$value$ $list$ which is having all the transactions who has performed \rvmt{} on the same key $k_1$ and $vnext$ is having the information about next available version of the same key $k_1$. The \mvotm system consists of the following main methods: \emph{STM init()}, \emph{STM begin()}, $\npins{}$, $\npluk{}$, $\npdel{}$ and $\nptc{}$.
}

\vspace{1mm}
\noindent
\textbf{\tins{()}}: This is another optimization done in \emph{\ophmvotm{s}} to identify the early abort which prevents the work done by aborted transactions and saves time. The actual effect of the \tins{()} comes after the successful tryC method. First, t\_insert() searches the key $k$ in the local log, $txLog_i$ of $T_i$ at \Lineref{insert2}. If $k$ does not exist in the $txLog_i$ then it identifies the appropriate location ($pred$ and $curr$) of key $k$ using \bn{} and \rn{} (\Lineref{insert3}) in the \lsl of $B_k$ without acquiring any locks similar to \rvmt{} explained above. 

Next, it acquires the re-entrant locks on all the $pred$ and $curr$ keys in increasing order. After that, all the $pred$ and $curr$ keys are validated by \emph{tryC\_Validation} in \Lineref{insert5} as follows: (1) It does the \emph{rv\_Validation()} as explained above in the \rvmt{}. (2) If key $k$ exists in the $B_k.\lsl$ and let $n$ as a node of $k$. Then algorithm identifies the version of $n$, $ver_j$ which has the timestamp $j$ such that $j$ has the largest timestamp smaller than $i$ (timestamp of $T_i$) at \Lineref{trycv5}. If $max_{rvl}$ of $ver_j$ is greater than timestamp $i$ at \Lineref{trycv6} then it returns $Abort$ in \Lineref{insert6}.

\emph{tryC\_Validation()} in \tins{()} identifies the early abort of invalid transaction. The advantage of doing the early validation to save the significant computation of long running transaction which will abort in the future. Consider \figref{insertadv} where two transaction $T_1$ and $T_2$ working on key $k_5$. In \figref{insertadv} (a), $T_1$ aborts in tryC (delayed validation) because higher timestamp $T_2$ committed. But in \figref{insertadv} (b), $T_1$ validates the t\_insert() instantly by looking into the $max_{rvl}$ of $k_5$ as shown in \figref{insertadv} (c) and save its computation and returns abort. 

\begin{algorithm}
	\scriptsize
	\caption{\emph{t\_insert():} Actual insertion happens in the tryC.} \label{algo:insert} 	
	\begin{algorithmic}[1]
		\makeatletter\setcounter{ALG@line}{32}\makeatother
		\Procedure{$t\_insert()$}{}\label{lin:insert1}		
		\If{($k \notin \llog_i$)} \label{lin:insert2}
		\State Search in \lsl to identify the $preds[]$ and $currs[]$ for $k$ using \bn and \rn in bucket $B_k$. \label{lin:insert3}
		\State Acquire the locks on $preds[]$ and $currs[]$ in increasing order. \label{lin:insert4}
		\If{($!tryC\_Validation()$)} \label{lin:insert5}
		\State return $Abort$.\label{lin:insert6} \Comment{Release the locks}
		\EndIf 
		\State Release the locks.
		\Else \label{lin:insert7}
		\State Update the local log.  \label{lin:insert8}
		\EndIf 
		\EndProcedure
	\end{algorithmic}
\end{algorithm}

\begin{figure}
	\captionsetup{justification=centering}
	\centerline{\scalebox{0.37}{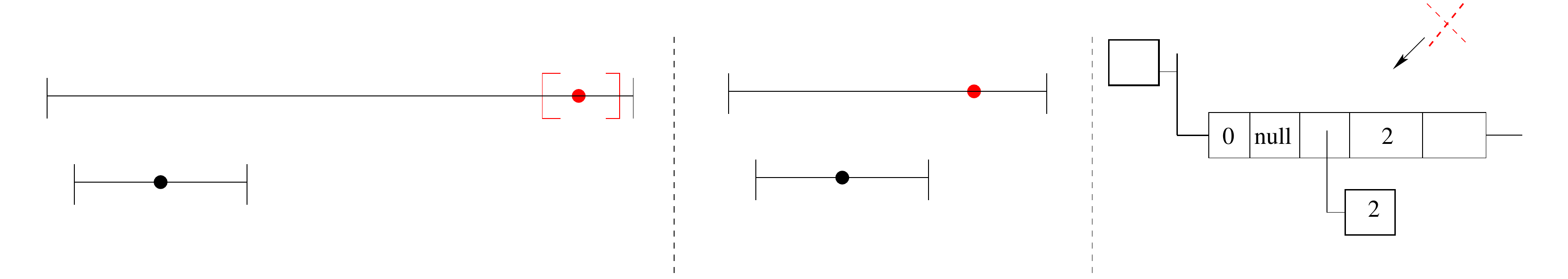}}
	\caption{Advantage of early validation in t\_insert()}
	\label{fig:insertadv}
\end{figure}

\vspace{1mm}
\noindent
\textbf{\upmt{s}}: It can be either $\tins(ht, k, v)$ or $\tdel(ht, k, v)$.
Both the \mth{s} create a version corresponding to the key $k$. The actual effect of \tins{} and \tdel{} in shared memory will take place in \tryc{}. \algoref{mtryc} represents the high level overview of \tryc{}. 

Initially, to avoid deadlocks, the algorithm sorts all the $keys$ in increasing order which are present in the local log, $\llog_i$. In \tryc{}, $\llog_i$ consists of \upmt{s} (\tins{} or \tdel{}) only. For all the \upmt{s} ($opn_i$) it searches the key $k$ in the shared memory corresponding to the bucket $B_k$. It identifies the appropriate location ($pred$ and $curr$) of key $k$ using \bn{} and \rn{} (\Lineref{mtryc3}) in the \lsl of $B_k$ without acquiring any locks similar to \rvmt{} explained above.

Next, it acquires the re-entrant locks on all the $pred$ and $curr$ keys in increasing order. After that, all the $pred$ and $curr$ keys are validated by \emph{tryC\_Validation} in \Lineref{mtryc5} as explained in t\_insert().


\begin{algorithm}
	
	\scriptsize
	\caption{\emph{tryC($T_i$)}: Validate the \upmt{s} of the transaction and then commit.}
	\setlength{\multicolsep}{0pt}
	\label{algo:mtryc}
	\begin{algorithmic}[1]
		\makeatletter\setcounter{ALG@line}{44}\makeatother
		\Procedure{$tryC{(T_i)}$}{}
		\State /*Operation name ($opn$) which could be either \tins or \tdel*/
		\State /*Sort the $keys$ of $\llog_i$ in increasing order.*/
		\ForAll{($opn_i$ $\in$ $\llog_i$)} \label{lin:mtryc1}
		\If{(($opn_i$ == \tins{}) $||$ ($opn_i$ == \tdel{}))}\label{lin:mtryc2}
		\State Search in $\lsl$ to identify the $preds[]$ and $currs[]$ for $k$ using \bn and \rn in bucket $B_k$. \label{lin:mtryc3}
		\State Acquire the locks on $preds[]$ and $currs[]$ in increasing order. \label{lin:mtryc4}
		\If{($!tryC\_Validation()$)} \label{lin:mtryc5}
		\State return $Abort$.\label{lin:mtryc6} \Comment{Release the locks}
		\EndIf\label{lin:mtryc7}
		\EndIf\label{lin:mtryc8}
		\EndFor\label{lin:mtryc9}
		\ForAll{($opn_i$ $\in$ $\llog_i$)}\label{lin:mtryc10}
		\State $intraTransValidation()$  modifies the $preds[]$ and $currs[]$ of current operation which would have been \blank{.7cm} updated by the previous operation of the same transaction.\label{lin:mtryc11}
		\If{(($opn_i$ == \tins{}) \&\& ($k  ~ \notin ~ B_k.\lsl$))}\label{lin:mtryc12}
		\State Create new node $n$ with $k$ as: $\langle$ \emph{key = k, lock = false, marked = false, vl = ver, nnext = $\phi$} $\rangle$. \label{lin:mtryc13}
		\State Create two versions $ver$ as: $\langle$ \emph{ts=0, val=null, rvl=$\phi$, $max_{rvl}$ = $\phi$, vnext=$i$} $\rangle$ for $T_0$ and $\langle$ \emph{ts=i, val=v, \blank{1.1cm}rvl=$\phi$, $max_{rvl}$ = $\phi$, vnext=$\phi$} $\rangle$ for $T_i$.
		\State Insert node $n$ into $B_k.\lsl$ such that it is accessible via \rn{} as well as \bn{} \Comment{$lock$ sets $true$}.  \label{lin:mmtryc14}
		\ElsIf{($opn_i$ == \tins{})}\label{lin:mtryc14}
		\State Add the version $ver$ as: $\langle$ \emph{ts=i, val=v, rvl=$\phi$, $max_{rvl}$=$\phi$, vnext=$\phi$} $\rangle$ into $B_k.\lsl$ such that it is \blank{1.1cm} accessible via \rn as well as \bn.\label{lin:mtryc15}
		\EndIf\label{lin:mtryc16}
		\If{($opn_i$ == \tdel{})}\label{lin:mtryc17}
		\State Add the version $ver$ as: $\langle$ \emph{ts=i, val=null, rvl=$\phi$, $max_{rvl}$=$\phi$, vnext=$\phi$} $\rangle$ into $B_k.\lsl$ such that \blank{1.1cm}it is accessible only via \rn.\label{lin:mtryc18}
		\EndIf\label{lin:mtryc19}
		\State Update the $preds[]$ and $currs[]$ of $opn_i$ in $\llog_i$.\label{lin:mtryc20}
		
		\EndFor
		\State Release the locks; return $Commit$.\label{lin:mtryc21}
		\EndProcedure
	\end{algorithmic}
\end{algorithm}

If \emph{tryC\_Validation} is successful then each \upmt{s} exist in $\llog_i$ will take the effect in the shared memory after doing the \emph{intraTransValidation()} in \Lineref{mtryc11}. If two $\upmt{s}$ of the same transaction have at least one common shared node among its recorded $pred$ and $curr$ keys, then the previous $\upmt{}$ effect may overwrite if the current $\upmt{}$ of $pred$ and $curr$ keys are not updated according to the updates are done by the previous $\upmt{}$. Thus to solve this we have \emph{intraTransValidation()} that modifies the $pred$ and $curr$ keys of current operation based on the previous operation in \Lineref{mtryc11}.


 
Next, we check if \upmt{} is \tins{} and $k$ is in $B_k.\lsl$. If $k$ is not in $B_k$, then create a new node $n$ for $k$ as $\langle key=k, lock=false, marked = false, vl=ver, nnext=\phi \rangle$. This node will have two versions $ver$ as $\langle ts=0, val=null, rvl=\phi, max_{rvl} = \phi, vnext=i \rangle$ for $T_0$ and $\langle ts=i, val=v, rvl=\phi, max_{rvl} = \phi, vnext=\phi \rangle$ for $T_i$. $T_i$ is creating a version with timestamp $0$ to ensure that \rvmt{s} of other transactions will never abort. For second version, $i$ is the timestamp of the transaction $T_i$ invoking this method; $marked$ field sets to false because the node is inserted in the \bn{}. $rvl$, $max_{rvl}$, and $vnext$ are initialized to empty values. We set the $val$ as $v$ and insert $n$ into $B_k.\lsl$ such that it is accessible via \rn{} as well as \bn{} and set the lock field to be $true$ (\linref{mmtryc14}). If $k$ is in $B_k.\lsl$ then, $k$ is the same as $\rc$ or $\bc$ or both. Let $n$ be the node of $k$ in $B_k.\lsl$.
Then, we create the version $ver$ as: $\langle ts=i, val=v, rvl=\phi, max_{rvl}=\phi, vnext=\phi \rangle$ and insert the version into $B_k.\lsl$ such that it is accessible via \rn{} as well as \bn{} (\linref{mtryc15}).

Subsequently, we check if \upmt{} is \tdel{} and $k$ is in $B_k.\lsl$. Let $n$ be the node of $k$ in $B_k.\lsl$.
Then create the version $ver$ as $\langle ts=i, val=null, rvl=\phi, max_{rvl}=\phi, vnext=\phi \rangle$ and insert the version into $B_k.\lsl$ such that it is accessible only via \rn{} (\linref{mtryc18}). 

Finally, at \Lineref{mtryc20} it updates the $pred$ and $curr$ of $opn_i$ in local log, $\llog_i$. At \Lineref{mtryc21} releases the locks on all the $pred$ and $curr$ in increasing order of keys to avoid deadlocks and return $Commit$.  

We illustrate the helping methods of \rvmt{}, t\_insert(), and \upmt{} in detail as follows:

\noindent
\textbf{rv\_Validation():} It is called by the \rvmt{}, t\_insert(), and \upmt{}. It identifies the conflicts among the concurrent methods of different transactions. Consider an example shown in \figref{mvostm9}, where two concurrent conflicting methods of different transactions are working on the same key $k_4$. Initially, at stage $s_1$ in \figref{mvostm9} (c) both the conflicting method optimistically (without acquiring locks) identify the same $pred$ and $curr$ keys for key $k_4$ from $B_k.\lsl$ in \figref{mvostm9} (a). At stage $s_2$ in \figref{mvostm9} (c), method $ins_1(ht, k_4, v_1)$ of transaction $T_1$ acquired the lock on $pred$ and $curr$ keys and inserted the node into $B_k.\lsl$ as shown in \figref{mvostm9} (b). After successful insertion by $T_1$, $pred$ and $curr$ have been changed for $lu_2(ht, k_4)$ at stage $s_3$ in \figref{mvostm9} (c). So, the above modified information is delivered by \emph{rv\_Validation} method at \Lineref{rvv1} when $(\bp.\bn \neq \bc)$ for $lu_2(ht, k_4)$. After that again it will find the new $pred$ and $curr$ for $lu_2(ht, k_4, v_1)$ and eventually it will commit. 
\begin{algorithm}
	\scriptsize
	\caption{\emph{rv\_Validation()}: Validate against the conflicting method of different transactions.}
	\setlength{\multicolsep}{0pt}
	\begin{algorithmic}[1]
		\makeatletter\setcounter{ALG@line}{72}\makeatother
		\Procedure{$rv\_validation{()}$}{}
		\If{$((\bp.marked) || (\bc.marked) ||(\bp.\bn) \neq \bc || (\rp.\rn) \neq {\rc})$}\label{lin:rvv1}
		\State return $false$. \label{lin:rvv2}
		\Else \label{lin:rvv3}
		\State return $true$. \label{lin:rvv4}
		\EndIf 
		\EndProcedure
	\end{algorithmic}
\end{algorithm}
\begin{figure}
	\captionsetup{justification=centering}
	\centerline{\scalebox{0.4}{\input{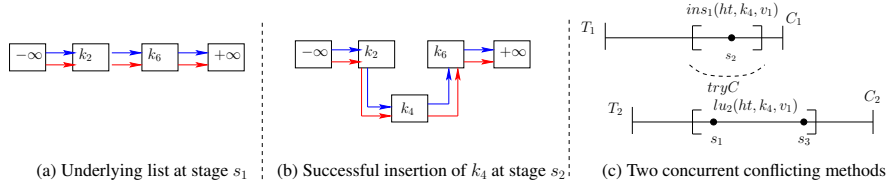}}}
	\caption{Illustration of rv\_Validation()}
	\label{fig:mvostm9}
\end{figure}

\begin{figure}[H] 
	\captionsetup{justification=centering}
	\centerline{\scalebox{0.45}{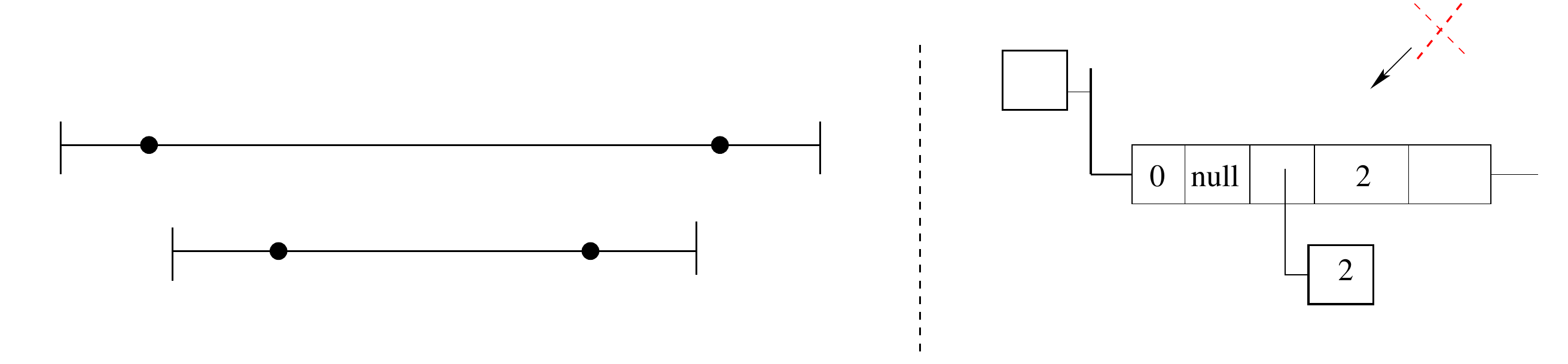}}
	\caption{Illustration of tryC\_Validation()}
	\label{fig:mvostm12}
\end{figure}

\noindent
\textbf{tryC\_Validation():} It is called by t\_insert(), and \upmt{} in \tryc{}. First, it does the \emph{rv\_Validation()} in \Lineref{trycv1}. If its successful and key $k$ exists in the $B_k.\lsl$ and let $n$ as a node of $k$. Then algorithm identifies the version of $n$, $ver_j$ which has the timestamp $j$ such that $j$ has the largest timestamp smaller than $i$ (timestamp of $T_i$) at \Lineref{trycv5}. If $max_{rvl}$ of $ver_j$ is greater than the timestamp of $i$ then the algorithm returns false (in \Lineref{trycv8}) and eventually, return $Abort$ in \Lineref{insert6} or \Lineref{mtryc6}. Consider an example as shown in \figref{mvostm12} (a), where second method $ins_1(ht, k_5)$ of transaction $T_1$ returns $Abort$ because higher timestamp of transaction $T_2$ is already present in the $max_{rvl}$ of version $T_0$ identified by $T_1$ in \figref{mvostm12} (b).

\begin{algorithm}[H]
	\scriptsize
	\caption{\emph{tryC\_Validation()}: It maintains the order among the transactions.}
	\setlength{\multicolsep}{0pt}
	\begin{algorithmic}[1]
		\makeatletter\setcounter{ALG@line}{79}\makeatother
		\Procedure{$tryC\_validation{()}$}{}
		\If{($!rv\_Validation()$)}\label{lin:trycv1}
		\State Release the locks and retry.\label{lin:trycv2}
		\EndIf\label{lin:trycv3}
		\If{(k $\in$ $B_k.\lsl$)}\label{lin:trycv4}
		\State Identify the version $ver_j$ with $ts=j$ such that $j$ is the largest timestamp smaller than $i$.\label{lin:trycv5}
		\If{($ver_j.max_{rvl}$ $>$ $i$)}\label{lin:trycv6}
		\State return $false$.\label{lin:trycv8}
		\EndIf \label{lin:trycv9}
		\EndIf\label{lin:trycv11}
		\State return $true$.\label{lin:trycv12}
		\EndProcedure
	\end{algorithmic}
\end{algorithm}
\begin{algorithm}[H]
	\label{alg:intra} 
	\scriptsize
	\caption{\emph{intraTransValidation()}: Help the upcoming method of the same transaction.}
	\setlength{\multicolsep}{0pt}
	\begin{algorithmic}[1]
		\makeatletter\setcounter{ALG@line}{91}\makeatother
		\Procedure{$intraTransValidation{()}$}{}
		\If{$((\bp.marked) || (\bp.\bn \neq \bc ))$} \label{lin:intra1}
		\If{($opn_k$ == Insert)}\label{lin:intra2}
		\State /*Modify the pred of current transaction $T_i$ with the help of previous transaction $T_k$*/
		\State $\bp_{i}$ = $\bp_{k}$.\bn.\label{lin:intra3} \Comment{Set the $T_i$ \bp{} as $T_k$ \bc{}}
		\Else\label{lin:intra4}
		\State $\bp_{i}$ = $\bp_{k}$.\label{lin:intra5} \Comment{Set the $T_i$ \bp{} as $T_k$ \bp{}}
		\EndIf\label{lin:intra6}
		\EndIf \label{lin:intra9}
		\If{(\rp.\rn $\neq$ \rc)}\label{lin:intra7}
		\State $\rp_{i}$ = $\rp_{k}.\rn$.\label{lin:intra8} \Comment{Set the $T_i$ \rp{} as $T_k$ \rc}
		\EndIf\label{lin:intra10}
		\EndProcedure
	\end{algorithmic}
\end{algorithm}

\noindent
\textbf{intraTransValidation():} It is called by \upmt{} in \tryc{}. If two $\upmt{s}$ of the same transaction have at least one common shared node among its recorded $pred$ and $curr$ keys, then the previous $\upmt{}$ effect may overwrite if the current $\upmt{}$ of $pred$ and $curr$ keys are not updated according to the updates done by the previous $\upmt{}$. Thus to solve this we have \emph{intraTransValidation()} that modifies the $pred$ and $curr$ keys of current operation based on the previous operation from \Lineref{intra1} to \Lineref{intra10}. Consider an example as shown in \figref{mvostm11}, where two \upmt{s} of transaction $T_1$ are $ins_{11}(ht, k_4, v_1)$ and $ins_{12}(ht, k_6, v_2)$ in \figref{mvostm11} (c). At stage $s_1$ in \figref{mvostm11} (c) both the \upmt{s} identify the same $pred$ and $curr$ from underlying DS as $B_k.\lsl$ shown in \figref{mvostm11} (a). After the successful insertion done by first \upmt{} at stage $s_2$ in \figref{mvostm11} (c), key $k_4$ is part of $B_k.\lsl$ (\figref{mvostm11} (b)). At stage $s_3$ in \figref{mvostm11} (c), $ins_{12}(ht, k_6, v_2)$ identified $(\bp.\bn \neq \bc)$ in \emph{intraTransValidation()} at \Lineref{intra1}. So it updates the $\bp$ in \Lineref{intra3} for correct updation in $B_k.\lsl$. 




\begin{figure}[H]
	\captionsetup{justification=centering}
	\centerline{\scalebox{0.43}{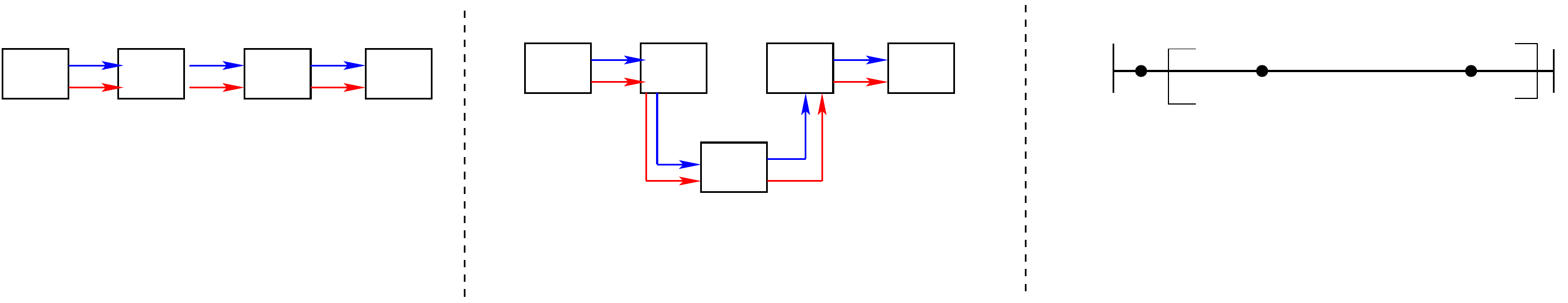}}
	\caption{Illustration of intraTransValidation()}
	\label{fig:mvostm11}
\end{figure}

\noindent

\section{Correctness of \opmvotm}
\label{sec:cmvostm}
In this section, we will prove that our implementation satisfies opacity. Consider the history $H$ generated by \emph{OPT-MVOSTM} algorithm. Recall that only the \emph{t\_begin}, \rvmt{}, t\_insert(), \upmt{} (or $tryC$) access shared memory. 

Note that $H$ is not necessarily sequential: the transactional \mth{s} can execute in an overlapping manner. To reason about correctness, we have to prove $H$ is opaque. Since we defined opacity for histories which are sequential, we order all the overlapping \mth{s} in $H$ to get an equivalent sequential history. We then show that this resulting sequential history satisfies \mth{}.

We order overlapping \mth{s} of $H$ as follows: (1) two overlapping \emph{t\_begin} \mth{s} based on the order in which they obtain lock over the $counter$; (2) two \rvmt{s} accessing the same key $k$ by their order of unlocking over $\langle \bp, \rp, \rc, \\\bc \rangle$ of $k$; (3) an \rvmt{} $rvm_i(k)$ and a $t\_insert_j()$, of a transaction $T_j$ accessing the same key $k$, are ordered by their order of unlocking over $\langle \bp, \rp, \\\rc, \bc \rangle$ of $k$; (4) an \rvmt{} $rvm_i(k)$ and a $\tryc_j$, of a transaction $T_j$ which has written to $k$, are similarly ordered by their order of unlocking over $\langle \bp, \rp, \rc, \bc \rangle$ of $k$; (5) two t\_insert{()} methods accessing the same key $k$ by their order of unlocking over $\langle \bp, \rp, \rc, \bc \rangle$ of $k$; (6) a $t\_insert_i()$ and a $\tryc_j$, of a transaction $T_j$ which has written to $k$, are similarly ordered by their order of unlocking over $\langle \bp, \rp, \rc, \bc \rangle$ of $k$; (7) similarly, two \tryc{} \mth{s} based on the order in which they unlock over $\langle \bp, \rp, \rc, \bc \rangle$ of same key $k$.

Combining the real-time order of events with above-mentioned order, we obtain a partial order which we denote as \emph{$\locko_H$}. (It is a partial order since it does not order overlapping \rvmt{s} on different $keys$ or an overlapping \rvmt{} and a \tryc{} which do not access any common $key$).

In order for $H$ to be sequential, all its \mth{s} must be ordered. Let $\alpha$ be a total order or \emph{linearization} of \mth{s} of $H$ such that when this order is applied to $H$, it is sequential. We denote the resulting history as $H^{\alpha} = \seq{\alpha}{H}$. 
We now argue about the \validity{} of histories generated by the algorithm. 
\begin{lemma}
	\label{lem:histvalid}
	Consider a history $H$ generated by the \opmvotm algorithm. Let $\alpha$ be a linearization of $H$ which respects $\locko_H$, i.e. $\locko_H \subseteq \alpha$. Then $H^{\alpha} = \seq{\alpha}{H}$ is \valid. 
\end{lemma}

\begin{proof}
	Consider a successful \rvmt{} $rvm_i(k)$ that returns value $v$. The \rvmt{} first obtains the lock on $\langle \bp, \rp, \rc, \bc \rangle$ of key $k$. Thus the value $v$ returned by the \rvmt{} must have already been stored in $k$'s version list by a transaction, say $T_j$ when it successfully returned OK from its \tryc{} \mth{}. For this to have occurred, $T_j$ must have successfully locked and released $\langle \bp, \rp, \rc, \\\bc \rangle$ of $k$ prior to $T_i$'s locking \mth. Thus from the definition of $\locko_H$, we get that $\tryc_j(ok)$ occurs before $rvm_i(k,v)$ which also holds in $\alpha$. 
	
\end{proof}

It can be seen that for proving correctness, any linearization of a history $H$ is sufficient as long as the linearization respects $\locko_H$. The following lemma formalizes this intuition, 

\begin{lemma}
	\label{lem:histseq}
	Consider a history $H$. Let $\alpha$ and $\beta$ be two linearizations of $H$ such that both of them respect $\locko_H$, i.e. $\locko_H \subseteq \alpha$ and $\locko_H \subseteq \beta$. Then, $H^{\alpha} = \seq{\alpha}{H}$ is opaque if $H^{\beta} = \seq{\beta}{H}$ is opaque. 
\end{lemma}

\begin{proof} From \lemref{histvalid}, we get that both $H^{\alpha}$ and $H^{\beta}$ are \valid{} histories. Now let us consider each case \\
	\textbf{If:}  Assume that $H^{\alpha}$ is opaque. Then, we get that there exists a legal \tseq{} history $S$ that is equivalent to $\overline{H^{\alpha}}$. From the definition of $H^{\beta}$, we get that $\overline{H^{\alpha}}$ is equivalent to $\overline{H^{\beta}}$. Hence, $S$ is equivalent to $\overline{H^{\beta}}$ as well. We also have that, $\prec_{H^{\alpha}}^{RT} \subseteq \prec_{S}^{RT}$. From the definition of $\locko_H$, we get that $\prec_{H^{\alpha}}^{RT} = \prec_{\locko_H}^{RT} = \prec_{H^{\beta}}^{RT}$. This automatically implies that $\prec_{H^{\beta}}^{RT} \subseteq \prec_{S}^{RT}$. Thus $H^{\beta}$ is opaque as well. 
	
	~ \\
	\textbf{Only if:} This proof comes from symmetry since $H^{\alpha}$ and $H^{\beta}$ are not distinguishable. 
\end{proof}

This lemma shows that, given a history $H$, it is enough to consider one sequential history $H^{\alpha}$ that respects $\locko_H$ for proving correctness. If this history is opaque, then any other sequential history that respects $\locko_H$ is also opaque.

Consider a history $H$ generated by \opmvotm{} algorithm. We then generate a sequential history that respects $\locko_H$. For simplicity, we denote the resulting sequential history of \opmvotm{} as $H_{to}$. Let $T_i$ be a committed transaction in $H_{to}$ that writes to $k$ (i.e. it creates a new version of $k$). 

To prove the correctness, we now introduce some more notations. We define $\stl{T_i}{k}{H_{to}}$ as a committed transaction $T_j$ such that $T_j$ has the \emph{smallest timestamp larger (or stl)} than $T_i$ in $H_{to}$ that writes to $k$ in $H_{to}$. Similarly, we define $\lts{T_i}{k}{H_{to}}$ as a committed transaction $T_k$ such that $T_k$ has the \emph{largest timestamp smaller (or lts)} than $T_i$ that writes to $k$ in $H_{to}$. \cmnt{ We denote $\vli{ts}{x}{H_{to}}$, as the $v\_tuple$ in $x.vl$ after executing all the events in $H_{to}$ created by transaction $T_{ts}$. If no such $v\_tuple$ exists then it is nil. }Using these notations, we describe the following properties and lemmas on $H_{to}$,

\begin{property}
	\label{prop:uniq-ts}
	Every transaction $T_i$ is assigned a unique numeric timestamp $i$.
\end{property}

\begin{property}
	\label{prop:ts-inc}
	If a transaction $T_i$ begins after another transaction $T_j$ then $j < i$.
\end{property}

\begin{lemma}
	\label{lem:readsfrom}
	If a transaction $T_k$ looks up key $k_x$ from (a committed transaction) $T_j$ then $T_j$ is a committed transaction updating to $k_x$ with $j$ being the largest timestamp smaller than $k$. Formally, $T_j = \lts{T_k}{k_x}{H_{to}}$. 
\end{lemma}

\begin{proof}
    We prove it by contradiction. So, assume that transaction $T_k$ looks up key $k_x$ from $T_i$ that has committed before $T_j$ so, from \propref{ts-inc}, $i<k$ and $k<j$ i.e. $i$ is not largest timestamp smaller than $k$. But given statement in this lemma is $i<j<k$ which contradicts our assumption. Hence, $T_k$ looks up key $k_x$ from $T_j$ which is the largest timestamp smaller than $k$.
\end{proof}

\cmnt{
	\begin{definition}
		\label{def:fLP}
		Linearization Point (LP) is the first unlocking point of each successful method.
	\end{definition}
}

\begin{lemma}
	\label{lem:readswrite}
	Suppose a transaction $T_k$ looks up $k_x$ from (a committed transaction) $T_j$ in $H_{to}$, i.e. $\{\up_j(k_{x,j}, v), \rvm_k(k_{x,i}, v)\} \in \evts{H_{to}}$. Let $T_i$ be a committed transaction that updates to $k_x$, i.e. $\up_i(k_{x,i}, u) \in \evts{T_i}$. Then, the timestamp of $T_i$ is either less than $T_j$'s timestamp or greater than $T_k$'s timestamp, i.e. $i<j \oplus k<i$ (where $\oplus$ is XOR operator). 
\end{lemma}

\begin{proof}
	We will prove this by contradiction. Assume that $i<j \oplus k<i$ is not true. This implies that, $j<i<k$. But from the implementation of \rvmt{} and \tryc{} \mth{s}, we get that either transaction $T_i$ is aborted or $T_k$ looks up $k$ from $T_i$ in $H$. Since neither of them are true, we get that $j<i<k$ is not possible. Hence, $i<j \oplus k<i$. 
\end{proof}

To show that $H_{to}$ satisfies opacity, we use the graph characterization developed above in \secref{gcofo}. For the graph characterization, we use the version order defined using timestamps. Consider two committed transactions $T_i, T_j$ such that $i < j$. Suppose both the transactions write to key $k$. Then the versions created are ordered as $k_i \ll k_j$. We denote this version order on all the $keys$ created as $\ll_{to}$. Now consider the opacity graph of $H_{to}$ with version order as defined by $\ll_{to}$, $G_{to} = \opg{H_{to}}{\ll_{to}}$. In the following lemmas, we will prove that $G_{to}$ is acyclic.

\begin{lemma}
	\label{lem:edgeorder}
	All the edges in $G_{to} = \opg{H_{to}}{\ll_{to}}$ are in timestamp order, i.e. if there is an edge from $T_j$ to $T_i$ then the $j < i$. 
\end{lemma}

\begin{proof}
	To prove this, let us analyze the edges one by one, 
	\begin{itemize}
		\item \rt{} edges: If there is an \rt{} edge from $T_j$ to $T_i$, then $T_j$ terminated before $T_i$ started. Hence, from \propref{ts-inc} we get that $j < i$.
		
		\item \rvf{} edges: This follows directly from \lemref{readsfrom}.
		
		\item \mv{} edges: The \mv{} edges relate a committed transaction $T_k$ updates to a key $k$, $up_k(k,v)$; a successful \rvmt{} $rvm_j(k,u)$ belonging to a transaction $T_j$ looks up $k$ updated by a committed transaction $T_i$, $up_i(k, u)$. Transactions $T_i, T_k$ create new versions $k_i, k_k$ respectively. According to $\ll_{to}$, if $k_k \ll_{to} k_i$, then there is an edge from $T_k$ to $T_i$. From the definition of $\ll_{to}$ this automatically implies that $k < i$.
		
		On the other hand, if $k_i \ll_{to} k_k$ then there is an edge from $T_j$ to $T_k$. Thus, in this case, we get that $i < k$. Combining this with \lemref{readswrite}, we get that $j < k$.

		
	\end{itemize}
	Thus in all the cases, we have shown that if there is an edge from $T_j$ to $T_i$ then the $j < i$.
\end{proof}

\begin{theorem}
	\label{thm:to-opaque}
	Any history $H_{to}$ generated by \opmvotm is \opq. 
\end{theorem}

\begin{proof}
	From the definition of $H_{to}$ and \lemref{histvalid}, we get that $H_{to}$ is \valid. We show that $G_{to} = \opg{H_{to}}{\ll_{to}}$ is acyclic. We prove this by contradiction. Assume that $G_{to}$ contains a cycle of the form, $T_{c1} \rightarrow T_{c2} \rightarrow .. T_{cm} \rightarrow T_{c1}$. From \lemref{edgeorder} we get that, $c1 < c2 < ... < cm < c1$ which implies that $c1 < c1$. Hence, a contradiction. This implies that $G_{to}$ is acyclic. Thus from \thmref{opg}, we get that $H_{to}$ is opaque.
\end{proof}

Now, it is left to show that our algorithm is \emph{live}, i.e., under certain conditions, every operation eventually completes. We have to show that the transactions do not deadlock. This is because all the transactions lock all $\langle \bp, \rp, \rc, \bc \rangle$ of $keys$ in a predefined order. As discussed earlier, the STM system orders all $\langle \bp, \rp, \rc, \bc \rangle$ of $keys$. We denote this order as \aco and denote it as $\prec_{ao}$. Thus $k_1 \prec_{ao} k_2 \prec_{ao} ... \prec_{ao} k_n$. 

From \aco, we get the following property

\begin{property}
	\label{prop:aco}
	Suppose transaction $T_i$ accesses shared objects $p$ and $q$ in $H$. If $p$ is ordered before $q$ in \aco, then $lock(p)$ by transaction $T_i$ occurs before $lock(q)$. Formally, $(p \prec_{ao} q) \Leftrightarrow (lock(p) <_H lock(q))$. 
\end{property}

\begin{theorem}
	\label{thm:prog}
	\opmvotm with unbounded versions ensures that \rvmt{s} do not abort.
\end{theorem}
\begin{proof}
	This is self-explanatory with the help of \opmvotm{} algorithm because each $key$ is maintaining multiple versions in the case of unbounded versions. So \rvmt{} always finds a correct version to read it from. Thus, \rvmt{s} do not $abort$. 
\end{proof}

\section{Experimental Evaluation}
\label{sec:exp}
This section describes the experimental analysis of proposed \emph{\opmvotm{s}} with state-of-the-art STMs. We have three main goals in this section: (1) Analyze the performance benefits of the optimized multi-version object based STMs (or \emph{\opmvotm{s}}) over multi-version object based STMs (or \emph{\mvotm{s}}). (2) Evaluate the benefit of \emph{\opmvotm{s}} over the single-version object based STMs (or \emph{OSTMs}), and (3) Analyze the benefit of \emph{\opmvotm{s}} over  multi-version read-write STMs. We implement hash-table object and list object as \ophmvotm and \oplmvotm described in \secref{pcode}. We also consider the extension of this optimized multi-version object STMs to reduce memory usage. Specifically, we consider a variant that implements garbage collection with unbounded versions and another variant where the number of versions never exceeds a given threshold $K$ for both \emph{\ophmvotm{s}} and \emph{\oplmvotm{s}}. 


\vspace{1mm}
\noindent
\textbf{Experimental system:} The Experimental system is a large-scale 2-socket Intel(R) Xeon(R) CPU E5-2690 v4 @ 2.60GHz with 14 cores per socket and two hyper-threads (HTs) per core, for a total of 56 threads. Each core has a private 32KB L1 cache and 256 KB L2 cache (which is shared among HTs on that core). All cores on a socket share a 35MB L3 cache. The machine has 32GB of RAM and runs Ubuntu 16.04.2 LTS. All code was compiled with the GNU C++ compiler (G++) 5.4.0 with the build target x86\_64-Linux-gnu and compilation option -std=c++1x -O3.

\vspace{1mm}
\noindent
\textbf{STM implementations:} We have taken the implementation of NOrec-list \cite{Dalessandro+:NoRec:PPoPP:2010}, Boosting-list \cite{HerlihyKosk:Boosting:PPoPP:2008}, Trans-list \cite{ZhangDech:LockFreeTW:SPAA:2016},  ESTM \cite{Felber+:ElasticTrans:2017:jpdc}, and RWSTM directly from the TLDS framework\footnote{https://ucf-cs.github.io/tlds/}. And the implementation of MVOSTM \cite{Juyal+:MVOSTM:SSS:2018}, OSTM \cite{Peri+:OSTM:Netys:2018} and MVTO \cite{Kumar+:MVTO:ICDCN:2014} from our PDCRL library\footnote{https://github.com/PDCRL/}. We implemented our algorithms in C++. Each STM algorithm first creates N-threads, each thread, in turn, spawns a transaction. Each transaction exports \tbeg{}, \tins{}, \tlook{}, \tdel{} and \tryc{} methods as described in \secref{model}.

\vspace{1mm}
\noindent
\textbf{Methodology:\footnote{Code is available here: https://github.com/PDCRL/MVOSTM/OPT-MVOSTM}} We have considered three types of workloads: ($W1$) Li - Lookup intensive (90\% lookup, 8\% insert, and 2\% delete), ($W2$) Mi - Mid intensive (50\% lookup,  25\% insert, and 25\% delete), and ($W3$) Ui - Update intensive (10\% lookup,  45\% insert, and 45\% delete). The experiments are conducted by varying number of threads from 2 to 64 in power of 2, with 1000 keys randomly chosen. We assume that the \tab{} of \ophmvotm has five buckets and each of the bucket (or list in case of \oplmvotm) can have a maximum size of 1000 keys. Each transaction, in turn, executes 10 operations which include \tlook, \tdel{}, and \tins{} operations. 
We take an average over 10 results as the final result for each experiment.

\vspace{1mm}
\noindent
\textbf{Results:} \figref{ht-time} represents the performance benefit of all the variants of proposed optimized \mvotm with all variants of \mvotm for hash-table objects. It shows \ophkotm performs best among all the algorithms (OPT-HT-MVOSTM-GC, OPT-HT-MVOSTM, HT-KOSTM, HT-MVOSTM-GC, HT-MVOSTM) by a factor of 1.02, 1.11, 1.05, 1.07, 1.22 for workload W1, 1.06, 1.09, 1.07, 1.08, 1.15 for workload W2, and 1.01, 1.03, 1.02, 1.03, 1.08 for workload W3 respectively. Along with this, \figref{ht-abort} shows the abort count respective algorithms on workload W1, W2, and W3. This represents for less number of threads, the number of aborts are almost same for all the algorithms. But while increasing the number of threads, the number of aborts are least in \ophkotm as compare to others.   So, we compare the performance of \ophkotm with the state-of-the-art STMs as shown in \figref{ht-time-all}. \ophkotm
outperforms all the algorithms (HT-OSTM, ESTM, RWSTM, HT-MVTO, HT-KSTM) by a factor of 3.62, 3.95, 3.44, 2.75, 1.85 for W1,  1.44, 2.36, 4.45, 9.84, 7.42 for W2, and 2.11, 4.05, 7.84, 12.94, 10.70 for W3 respectively. The corresponding number of aborts are represented in \figref{ht-abort-all}. Number of aborts are minimum for \ophkotm as compare to other state-of-the-art STMs. Especially, the number of aborts for \ophkotm is almost negligible as compared to HT-OSTM on lookup-intensive workload (W1) because \ophkotm finds a correct version to looks up as shown in \figref{ht-abort-all} (a). 

The observation of optimized list based \mvotm{} is similar as 
optimized hash-table based \mvotm{}. \figref{list-time} represents the performance benefit of all the variants of proposed optimized \mvotm with all variants of \mvotm for list objects. It shows \oplkotm performs best among all the algorithms (OPT-list-MVOSTM-GC, OPT-list-MVOSTM, list-KOSTM, list-MVOSTM-GC, list-MVOSTM) by a factor of 1.14, 1.24, 1.21, 1.20, 1.35 for W1, 1.06, 1.07, 1.12, 1.13, 1.20 for W2, and  1.09, 1.19, 1.11, 1.17, 1.31 for W3 respectively. Along with this, \figref{list-abort} shows the minimum abort count  by \oplkotm as compare to other algorithms on workload W1, W2, and W3. Hence, we choose the best-proposed algorithm \oplkotm and compare with the state-of-the-art list based STMs. 

\begin{figure}[H]
	\centering
	\captionsetup{justification=centering}
	\includegraphics[width=12cm]{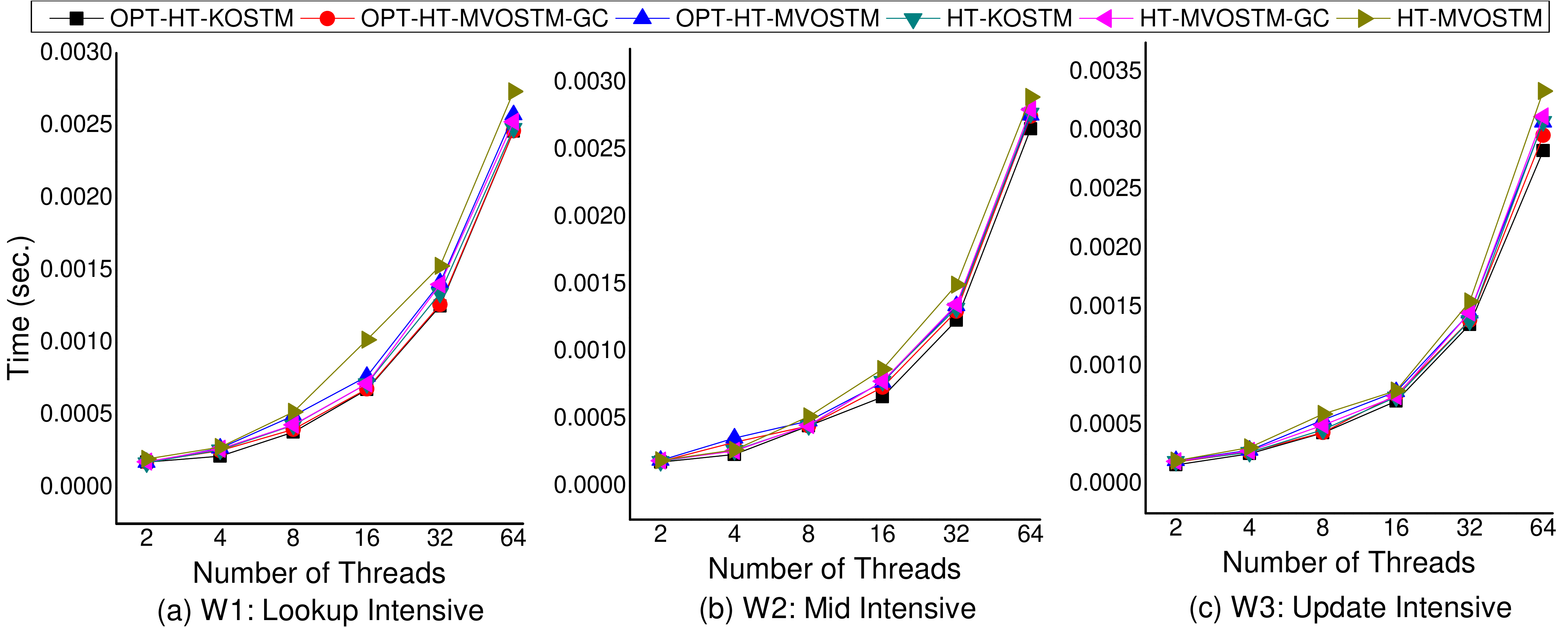}
	\centering
	\caption{Time comparison among variants of \emph{\ophmvotm{s}} and \emph{\hmvotm{s}} on hash-table}\label{fig:ht-time}
\end{figure}
\vspace{-.5cm}
\begin{figure}[H]
	\centering
	\captionsetup{justification=centering}
	\includegraphics[width=12cm]{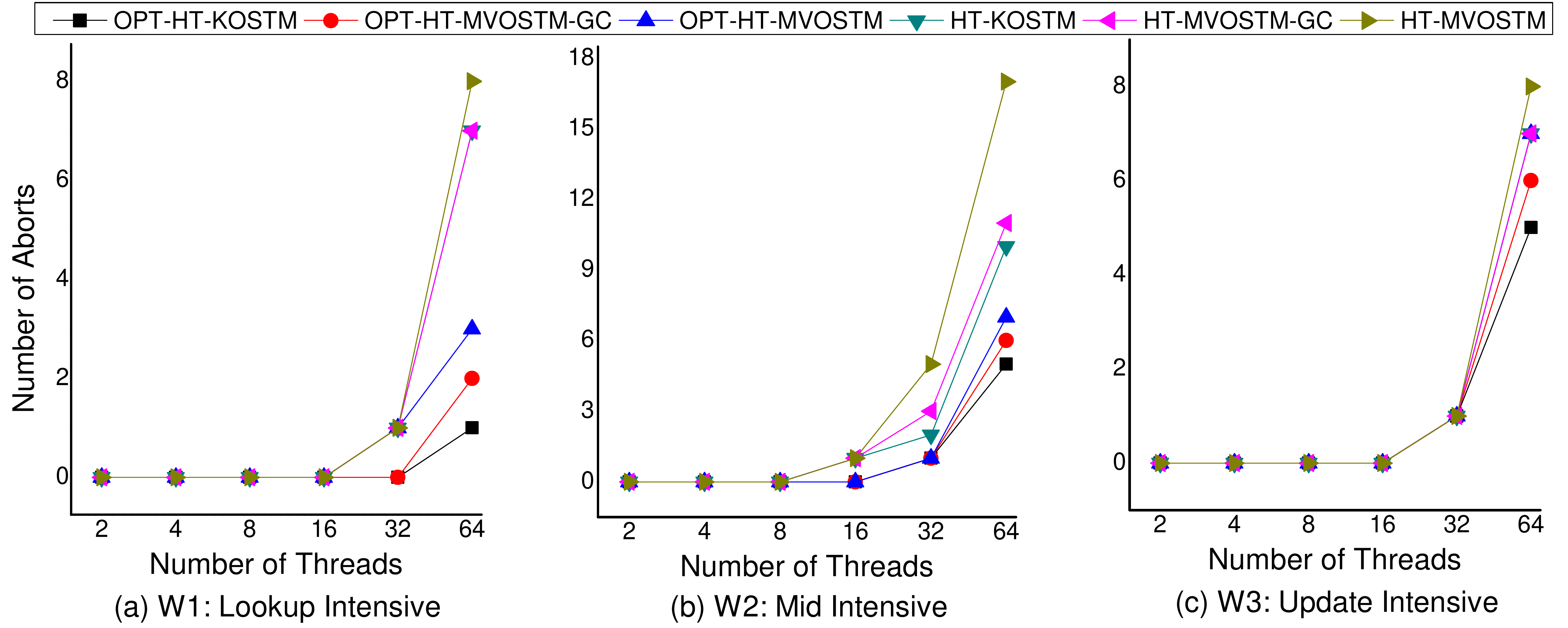}
	\centering
	\caption{Abort count among variants of \emph{\ophmvotm{s}} and \emph{\hmvotm{s}} on hash-table}\label{fig:ht-abort}
\end{figure}
\vspace{-.5cm}

\figref{list-time-all} represents \oplkotm
outperforms all the algorithms (list-OSTM, Trans-list, Boosting-list, NOrec-list, list-MVTO, list-KSTM)  by a factor of  2.56, 25.38, 23.57, 27.44, 13.34, 5.99 for W1, 1.51, 20.54, 24.27, 29.45, 24.89, 19.78 for W2, and 2.91, 32.88, 28.45, 40.89, 173.92, 124.89 for W3 respectively. Similarly, \figref{list-abort-all} depicts that \oplkotm obtained the least number of aborts as compare to others on the respective workloads.

\begin{figure}[H]
	\centering
	\captionsetup{justification=centering}
	\includegraphics[width=12cm]{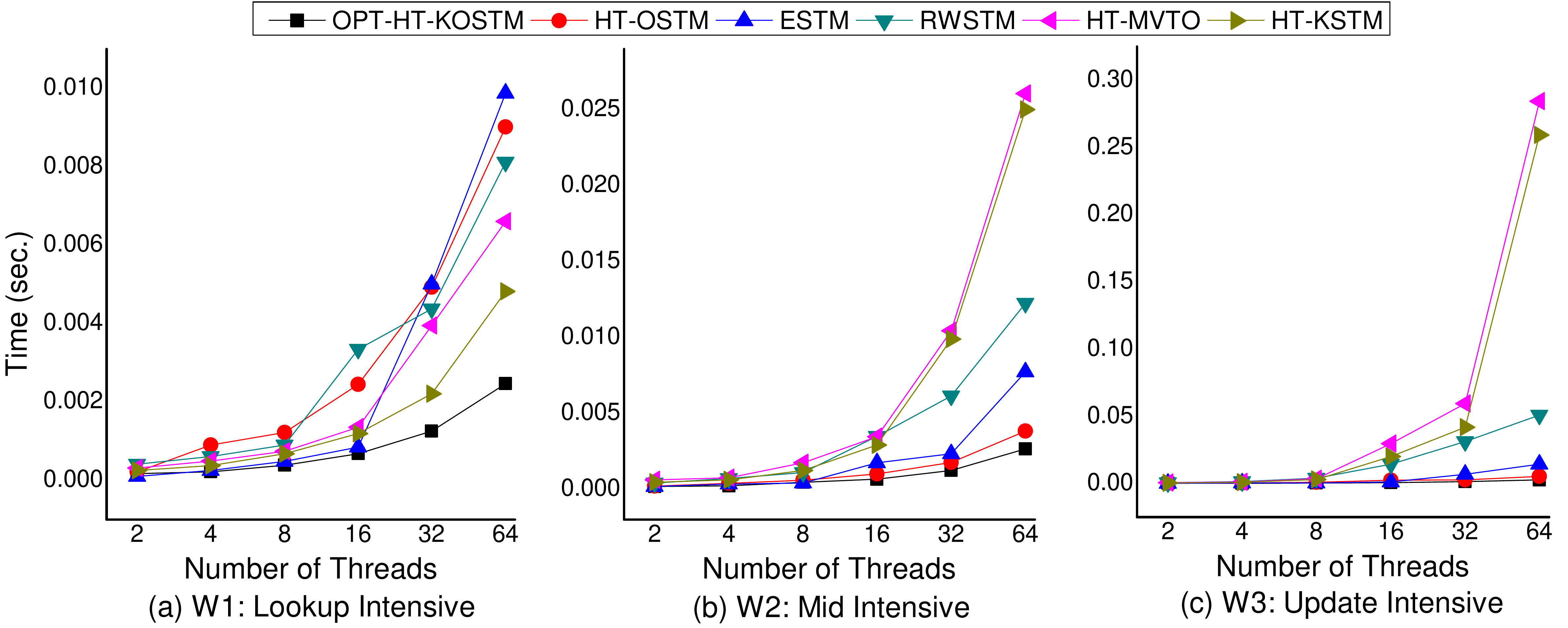}
	\centering
	\caption{Time comparison of \ophkotm and State-of-the-art hash-table based STMs}\label{fig:ht-time-all}
\end{figure}
\begin{figure}[H]
	\centering
	\captionsetup{justification=centering}
	\includegraphics[width=12cm]{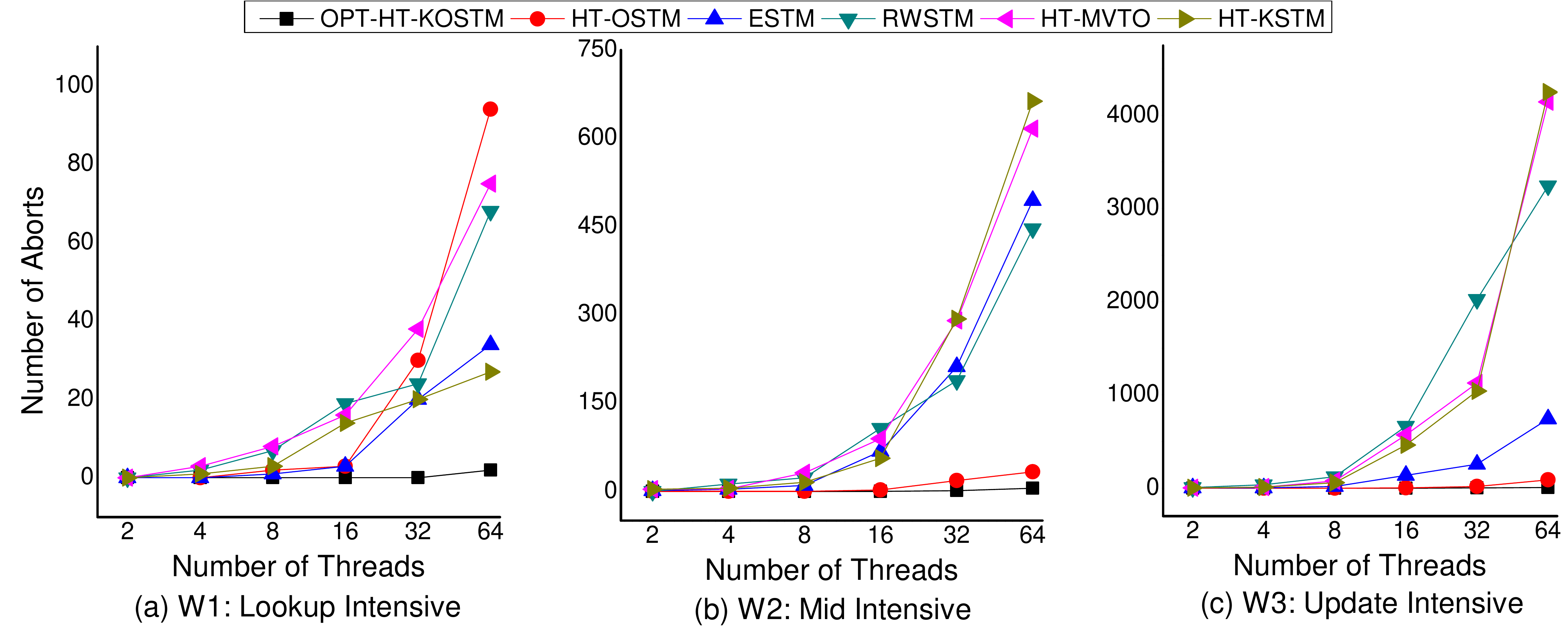}
	\centering
	\caption{Abort count of \ophkotm and State-of-the-art hash-table based STMs}\label{fig:ht-abort-all}
\end{figure}

\begin{figure}[H]
	\centering
	\captionsetup{justification=centering}
	\includegraphics[width=12cm]{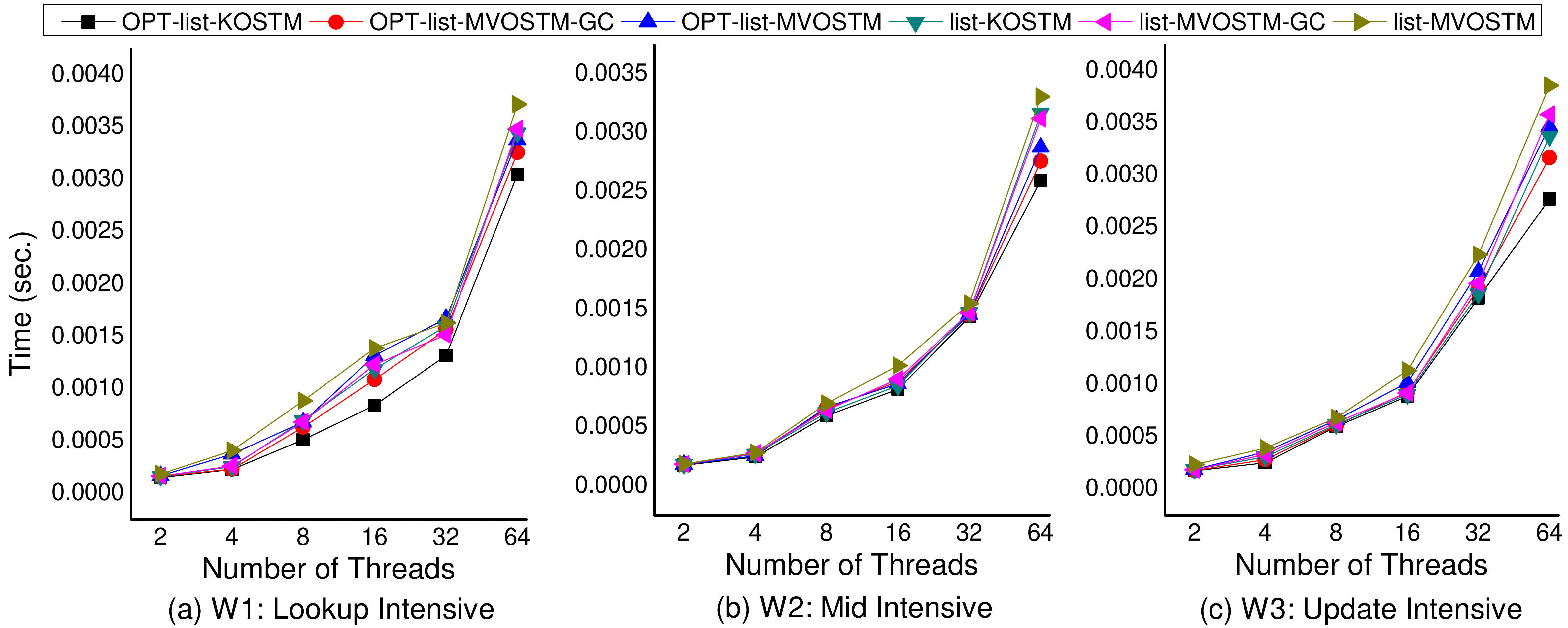}
	\centering
	\caption{Time comparison among variants of \emph{\oplmvotm{s}} and \emph{\lmvotm{s}} on list}\label{fig:list-time}
\end{figure}

\begin{figure}[H]
	\centering
	\captionsetup{justification=centering}
	\includegraphics[width=12cm]{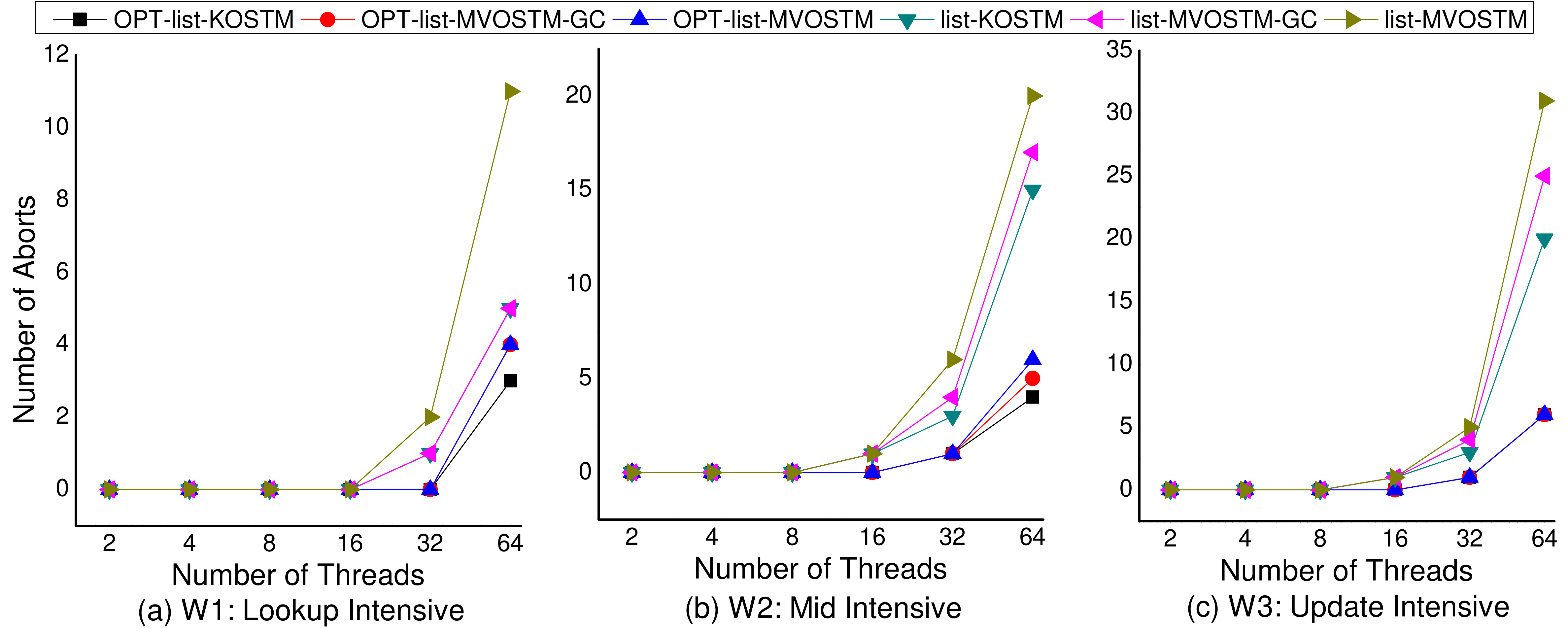}
	\centering
	\caption{Abort count among variants of \emph{\oplmvotm{s}} and \emph{\lmvotm{s}} on list}\label{fig:list-abort}
\end{figure}

\begin{figure}[H]
	\centering
	\captionsetup{justification=centering}
	\includegraphics[width=12cm]{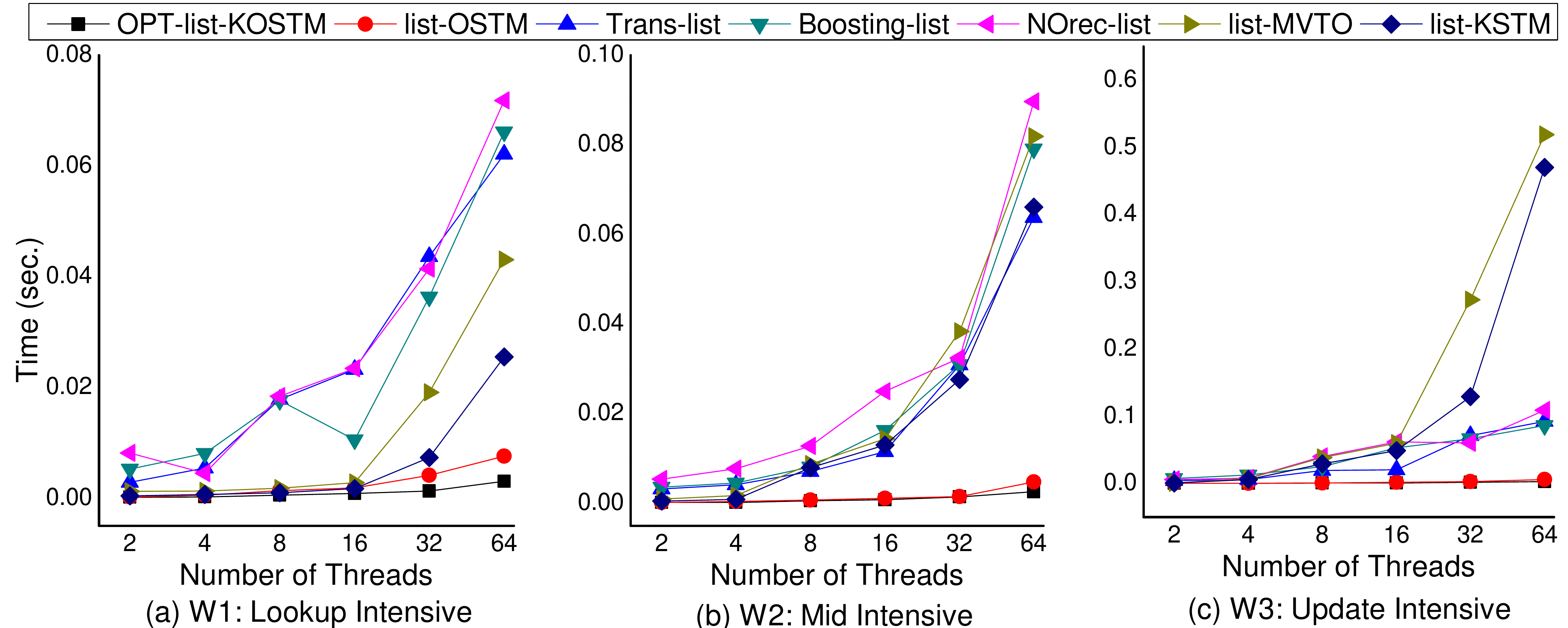}
	\centering
	\caption{Time comparison of \oplkotm and State-of-the-art list based STMs}\label{fig:list-time-all}
\end{figure}

\begin{figure}[H]
	\centering
	\captionsetup{justification=centering}
	\includegraphics[width=12cm]{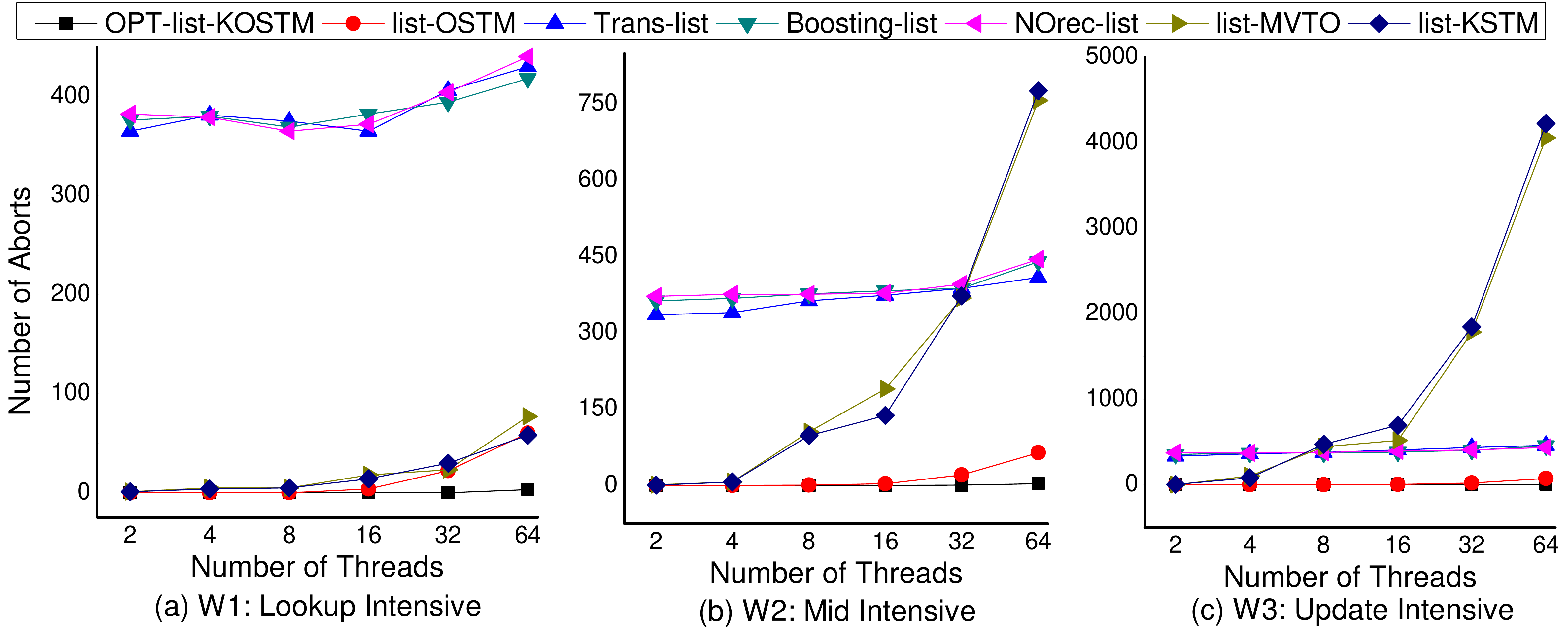}
	\centering
	\caption{Abort count of \oplkotm and State-of-the-art list based STMs}\label{fig:list-abort-all}
\end{figure}

\begin{figure}[H]
	\centering
	\captionsetup{justification=centering}
	\includegraphics[width=12cm]{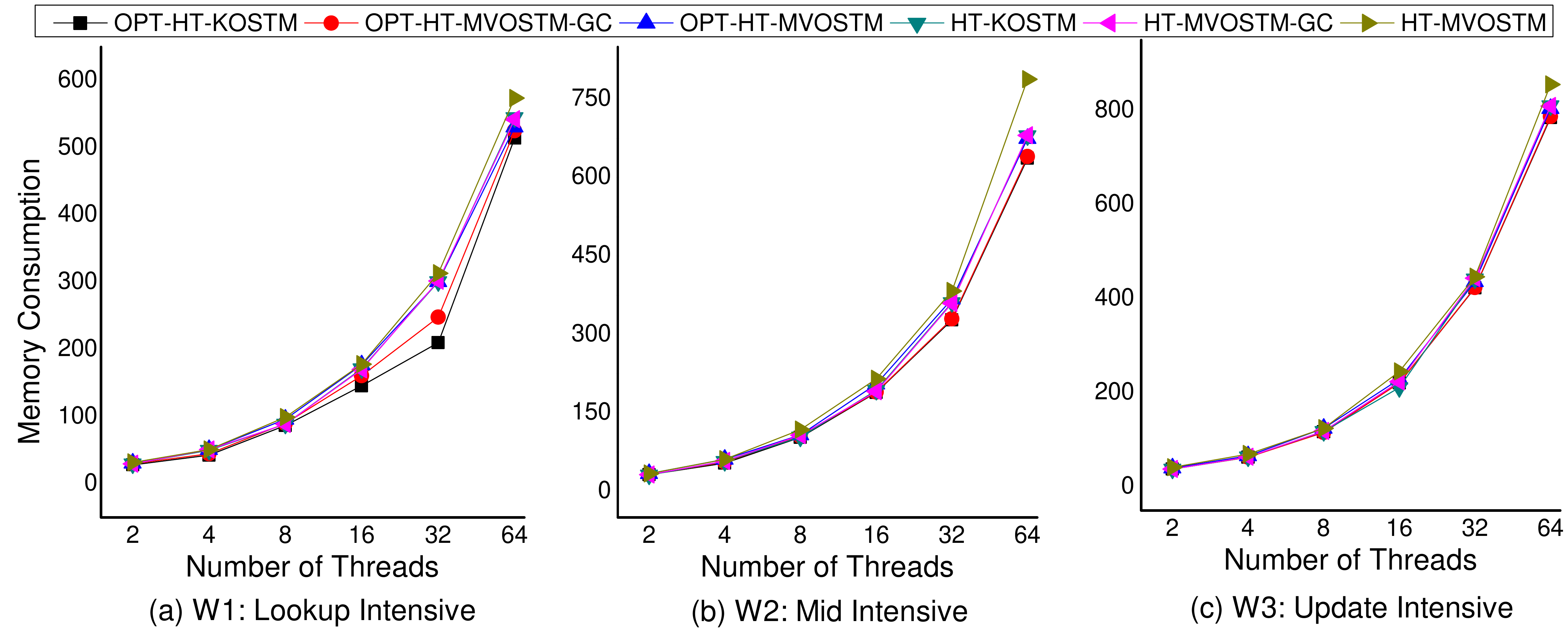}
	\centering
	\caption{Memory consumption among variants of \emph{\ophmvotm{s}} and \emph{\hmvotm{s}} on hash-table }\label{fig:ht-memory}
\end{figure}

\begin{figure}[H]
	\centering
	\captionsetup{justification=centering}
	\includegraphics[width=12cm]{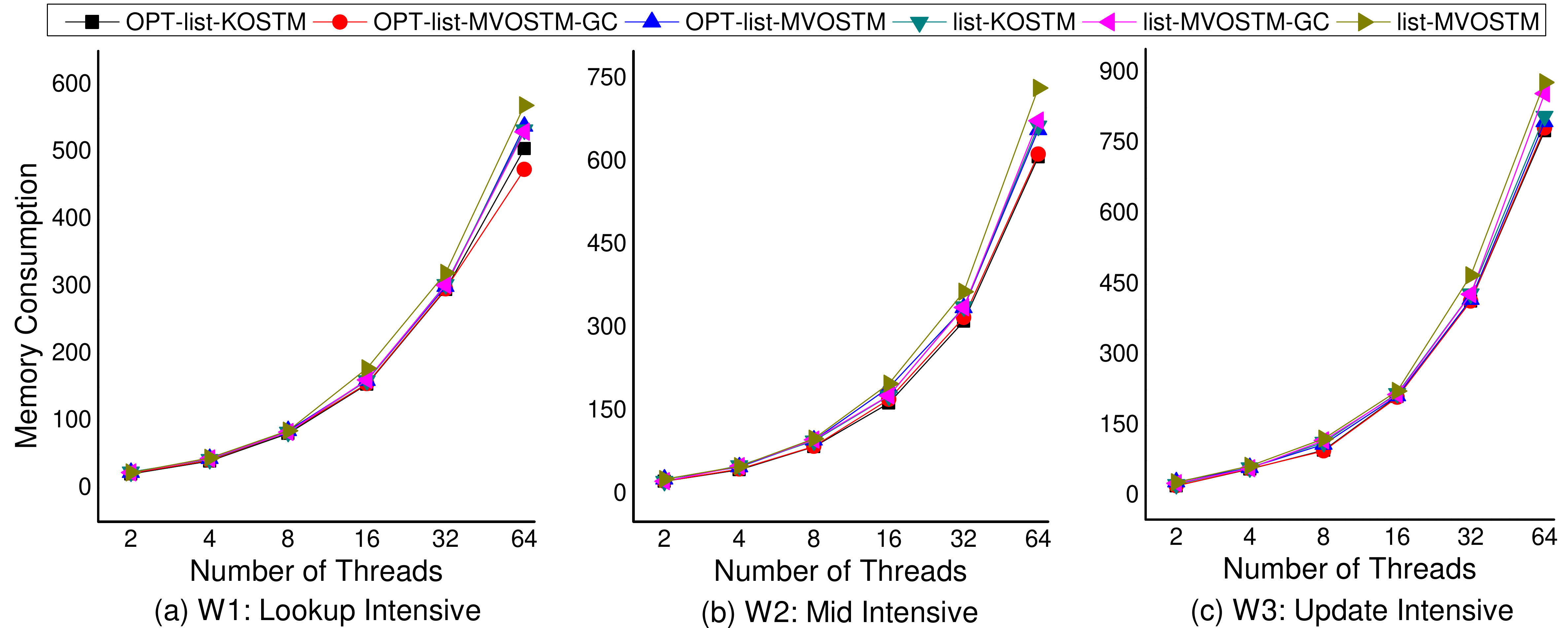}
	\centering
	\caption{Memory consumption among variants of \emph{\oplmvotm{s}} and \emph{\lmvotm{s}} on list}\label{fig:list-memory}
\end{figure}

\begin{figure}[H]
	\centering
	\captionsetup{justification=centering}
	\includegraphics[width=9cm, height=5cm]{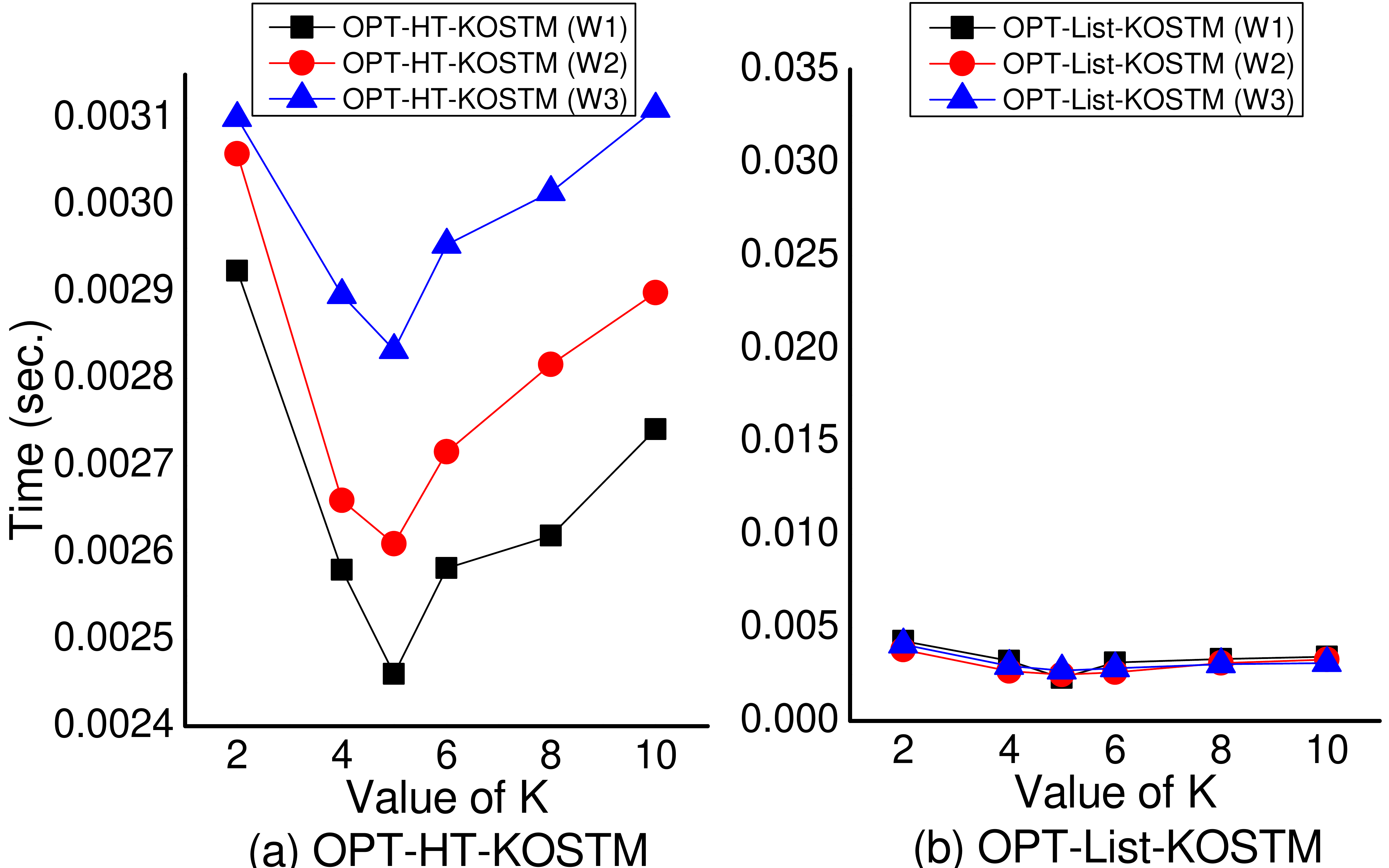}
	\centering
	\caption{Optimal Value of $K$ for \ophkotm and \oplkotm}\label{fig:optK}
\end{figure}

\cmnt{
\begin{figure}[H]
	\centering
	\captionsetup{justification=centering}
	\includegraphics[width=13cm]{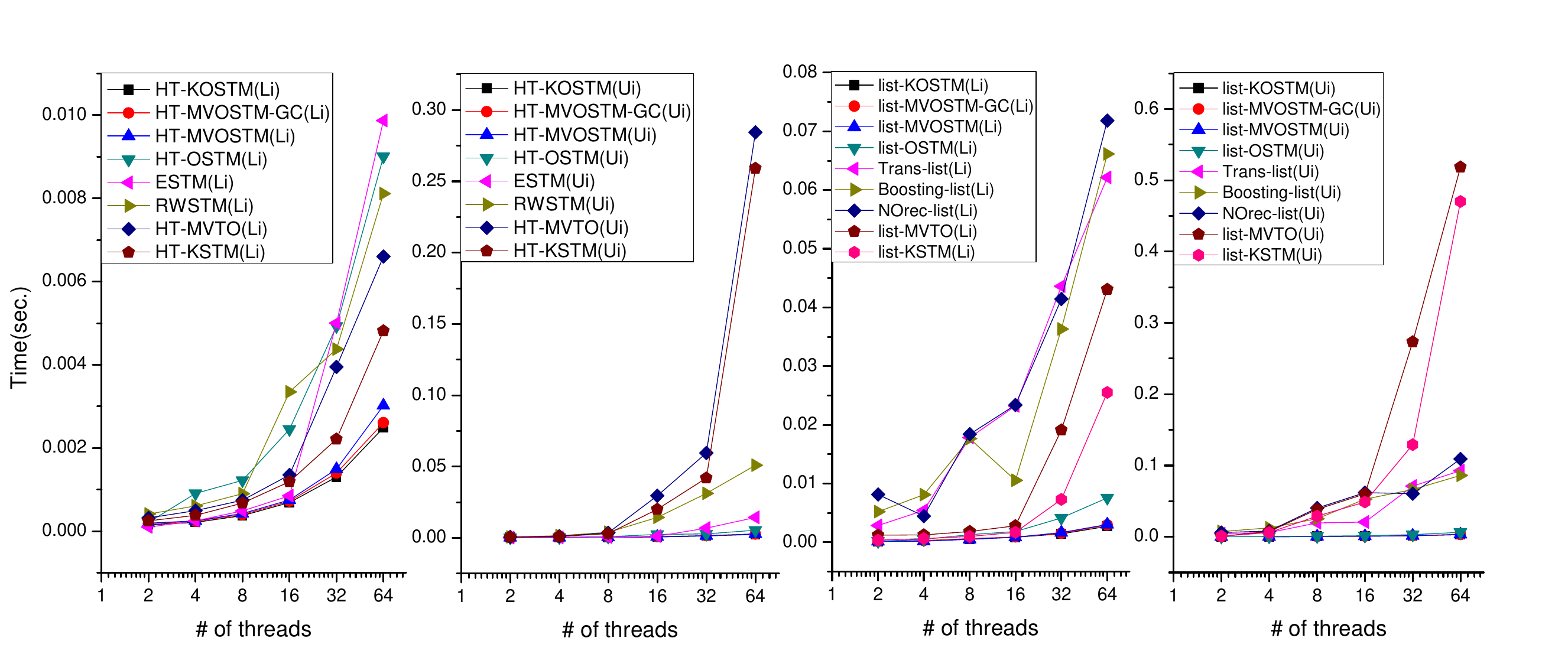}
	\centering
	\caption{Performance of \hmvotm and \lmvotm}\label{fig:htw1}
\end{figure}
\begin{figure}[H]
	\captionsetup{justification=centering}
	\includegraphics[width=13.5cm]{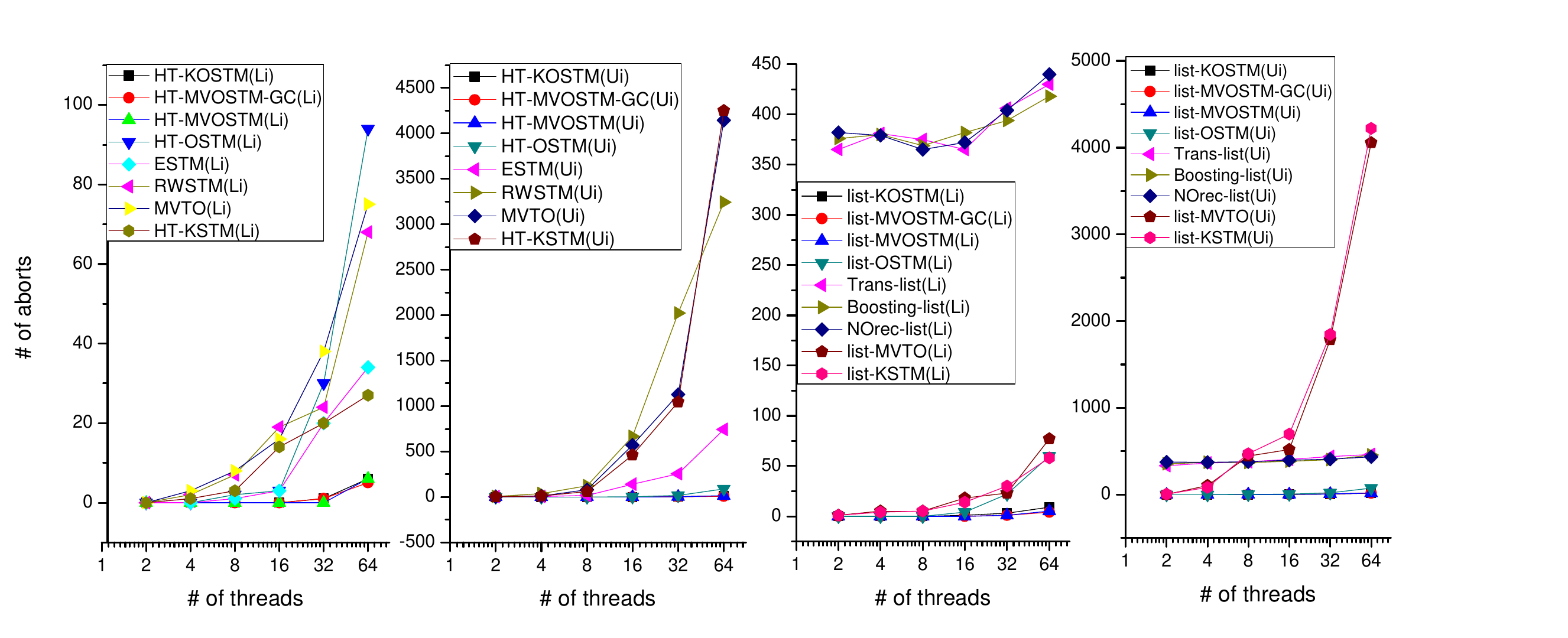}
	\centering
	\caption{Aborts of \hmvotm and \lmvotm}\label{fig:htw2}
\end{figure}
}

\noindent
\\
As explained in \secref{pcode}, for efficient memory utilization, we develop two variations of \opmvotm.
The first, \opmvotmgc, uses unbounded versions but performs garbage collection. \textbf{This is achieved by deleting non-latest versions whose timestamp is less than the timestamp of  the least live transaction.} \opmvotmgc gave a performance gain of 16\%  over \opmvotm without garbage collection in the best case which is on workload W1 with 64 number of threads. We did one more optimization in \opmvotmgc on the marked node exist in the \rn to make it search efficiently. \textbf{This is achieved by deleting a marked node from \rn whose $max_{rvl}$ of the last version is less than the timestamp of the least live transaction.}  
%
The second, 
\opkotm, keeps at most $K$  versions by replacing the oldest version when $(K+1)^{th}$ version is created by a current transaction as explained in \secref{pcode}. \opkotm shows a performance gain of 24\%  over \opmvotm without garbage collection in the best case which is on workload W1 with 64 number of threads. As \opkotm has a limited number of versions while \opmvotmgc can have infinite versions, the memory consumed by \opkotm is also less than \opmvotmgc. We have integrated these variations in both \tab{} based (\ophmvotmgc and \ophkotm) and linked-list based MVOSTMs (\oplmvotmgc and \oplkotm), we observed that these two variations increase the performance, concurrency and reduce the number of aborts as compared to OPT-MVOSTM which does not perform garbage collection.

\textbf{Memory Consumption by OPT-MVOSTM-GC and OPT-KOSTM:} As depicted above \opkotm performs better than \opmvotmgc. Continuing the comparison between the two variations of \opmvotm we chose another parameter as memory consumption. Here we test for the memory consumed by each variation algorithms in creating a version of a key. We count the total versions created, where creating a version increases the counter value by 1 and deleting a version decreases the counter value by 1. 
\figref{ht-memory} depicts the comparison of memory consumption by all the variants of proposed optimized \mvotm with all variants of \mvotm for hash-table objects. \ophkotm consumes minimum memory among all the algorithms (OPT-HT-MVOSTM-GC, OPT-HT-MVOSTM, HT-KOSTM, HT-MVOSTM-GC, HT-MVOSTM) by a factor of 1.07, 1.16, 1.15, 1.15, 1.21 for W1 , 1.01, 1.08, 1.06, 1.07, 1.19 for W2, and 1.01, 1.03, 1.02, 1.03, 1.08 for W3 respectively. Similarly, \figref{list-memory} depicts the comparison of memory consumption by all the variants of proposed optimized \mvotm with all variants of \mvotm for list objects. \oplkotm consumes minimum memory among all the algorithms
(OPT-list-MVOSTM-GC, OPT-list-MVOSTM, list-KOSTM, list-MVOSTM-GC, list-MVOSTM) by a factor of 1.01, 1.05, 1.05, 1.04, 1.11 for W1, 1.02, 1.1, 1.1, 1.11 1.19 for W2, and 1.01, 1.03, 1.05, 1.08, 1.13 for W3 respectively.

\textbf{Finite version \opmvotm (\opkotm):} 
To find the ideal value of $K$ such that performance as compared to \opmvotmgc does not degrade or can be increased, we perform experiments on all the workloads (W1, W2, and W3) for both (\ophkotm and \oplkotm). 
\figref{optK} (a) and (b) shows the best value of $K$ as 5 for \ophkotm and \oplkotm on all the workloads for both hash-table and list objects. 
\cmnt{
\begin{figure}
	\centering
	\captionsetup{justification=centering}
	\includegraphics[width=15cm]{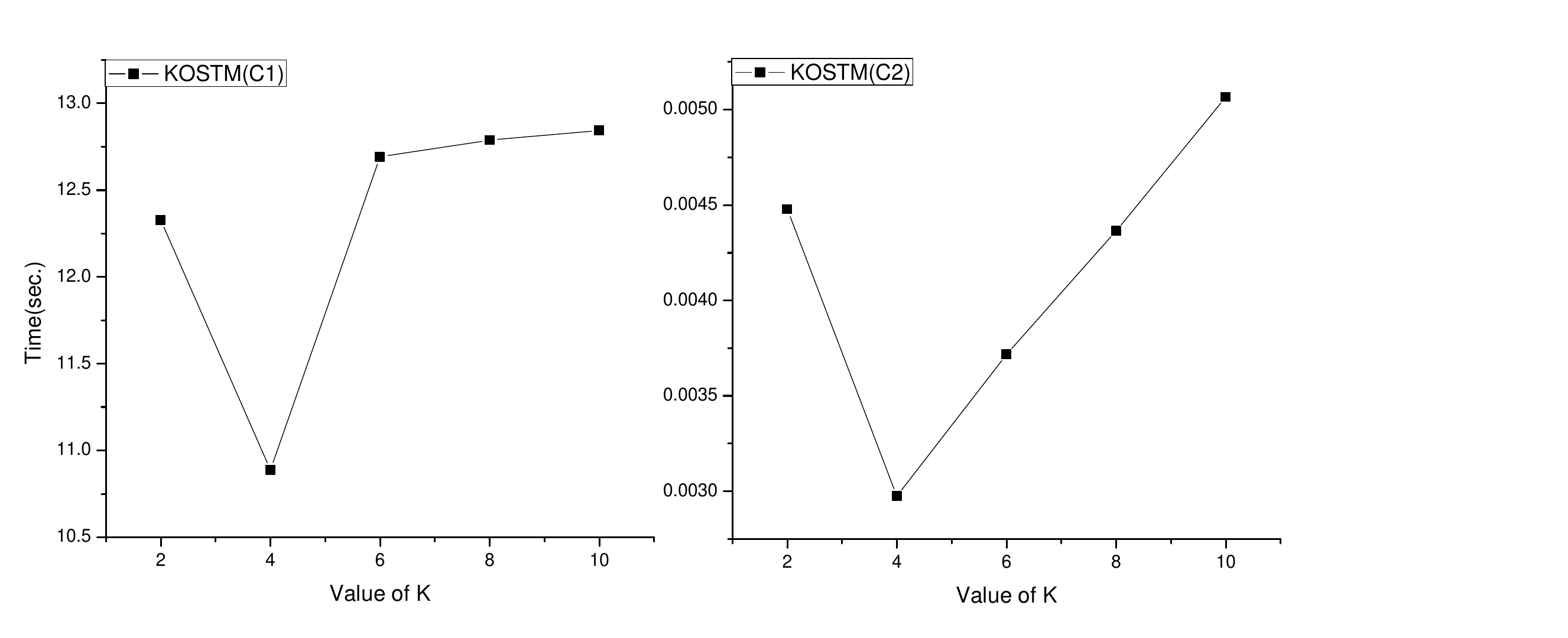}
	\centering
	\caption{Optimal value of K as 4}\label{fig:opmitalK}
\end{figure}
}

\vspace{1mm}
\noindent

\section{Conclusion}
\label{sec:confu}
With the rise of multi-core systems, concurrent programming becomes popular. Concurrent programming using multiple threads has become necessary to utilize all the cores present in the system effectively. But concurrent programming is usually challenging due to synchronization issues between the threads. 

In the past few years, several STMs have been proposed which address these synchronization issues and provide greater concurrency. STMs hide the synchronization and communication difficulties among the multiple threads from the programmer while ensuring correctness and hence making programming easy. Another advantage of STMs is that they facilitate compositionality of concurrent programs with great ease. Different concurrent operations that need to be composed to form a single atomic unit is achieved by encapsulating them in a single transaction.
 
In literature, most of the STMs are \rwtm{s} which export read and write operations. To improve the performance, a few researchers have proposed \otm{s} \cite{HerlihyKosk:Boosting:PPoPP:2008,Hassan+:OptBoost:PPoPP:2014, Peri+:OSTM:Netys:2018} which export higher level objects operation such as \tab{} insert, delete, and lookup etc. By leveraging the semantics of these higher level \op{s}, these STMs provide greater concurrency. On the other hand, it has been observed in STMs and databases that by storing multiple versions for each \tobj in case of \rwtm{s} provides greater concurrency \cite{perel+:2010:MultVer:PODC,Kumar+:MVTO:ICDCN:2014}.

\ignore{
	
\todo{SK: Remove this para}
So, we get inspired from literature and proposed a new STM as \mvotm which is the combination of both of these ideas (multi-versions in OSTMs). It provides greater concurrency and reduces number of abort while maintaining multiple versions corresponding to each key. \mvotm ensures compositionality by making the transactions atomic. In addition to that, we develop garbage collection for MVOSTM (MVOSTM-GC) to delete unwanted versions corresponding to the each keys to reduce traversal overhead. \mvotm satisfies \ccs{} as \emph{opacity}.

}

This paper proposed the notion of the optimized multi-version object based STMs (\opmvotm{s}) and compares their effectiveness with multi-version object based STMs (MVOSTMs), single-version object based STMs and multi-version read-write STMs. We find that \opmvotm provides a significant benefit over above-mentioned state-of-the-art STMs for different types of workloads. Specifically, we  have  evaluated the effectiveness of OPT-MVOSTM for the \tab{} and list data structure as \ophmvotm and \oplmvotm respectively. 

\ophmvotm and \oplmvotm use the unbounded number of versions for each key. To utilize the memory efficiently, we limit the number of versions and develop two variants for both \tab and list data structures: (1) A garbage collection method in \opmvotm to delete the unwanted versions of a key, denoted as \opmvotmgc. (2) Placing a limit of $K$ on the number of versions in \opmvotm, resulting in \opkotm. Both these variants (\opmvotmgc and \opkotm) gave a performance gain of over 16\% and 24\% over \opmvotm in the best case. \opkotm consumes minimum memory among all the variants of it. We represent \opmvotmgc in hash-table and list as \ophmvotmgc and \oplmvotmgc  respectively. Similarly, We represent \opkotm in hash-table and list as \ophkotm  and \oplkotm  respectively.

\ophkotm performs best among its variants and
outperforms state-of-the-art hash-table based STMs  (HT-OSTM, ESTM, RWSTM, HT-MVTO, HT-KSTM) by a factor of 3.62, 3.95, 3.44, 2.75, 1.85 for workload W1,  1.44, 2.36, 4.45, 9.84, 7.42 for workload W2, and 2.11, 4.05, 7.84, 12.94, 10.70 for workload W3 respectively. Similarly, \oplkotm performs best among its variants and outperforms state-of-the-art list based STMs (list-OSTM, Trans-list, Boosting-list, NOrec-list, list-MVTO, list-KSTM)  by a factor of  2.56, 25.38, 23.57, 27.44, 13.34, 5.99 for W1, 1.51, 20.54, 24.27, 29.45, 24.89, 19.78 for W2, and 2.91, 32.88, 28.45, 40.89, 173.92, 124.89 for W3 respectively.  We rigorously proved that \opmvotm{s} satisfy the correctness criteria as \opty.



\bibliography{citations}


\end{document}